\documentclass[12pt, reqno]{amsart}

\usepackage{amsmath,amssymb}
\usepackage{graphicx}
\usepackage{mathptmx}      
\usepackage{latexsym}
\usepackage{xcolor}
\usepackage{amsthm}
\usepackage{comment} 
\usepackage{appendix}

\usepackage{tikz}
\usetikzlibrary{shapes,arrows}
\usetikzlibrary{intersections}
\usepackage{pgfplots}
\pgfplotsset{compat = newest}  

\newcommand{\He}{\mathcal{H}}
\newcommand{\Se}{\mathcal{X}}
\newcommand{\se}{{\bf X}}
\newcommand{\Te}{\mathcal{T}}
\newcommand{\It}{I_{t_0}}
\newcommand{\deltahatone}{\hat{\delta}^1}
\newcommand{\deltahattwo}{\hat{\delta}^2}
\newcommand{\Xe}{\mathcal{X}}

\newtheorem{proposition}{Proposition}[section]
\newtheorem{theorem}[proposition]{Theorem}
\newtheorem{definition}[proposition]{Definition}
\newtheorem{corollary}[proposition]{Corollary}

\newtheorem{remark}[proposition]{Remark}

\begin{document}
\markboth{D. Crisci, K. Gajewsky and S.E. Ferrando}{Agent-Based Models for Two Stocks with Superhedging}


\title{Agent-Based Models for Two Stocks with Superhedging}



\author{D. Crisci}
\address[D. Crisci]{Department of Mathematics\\
Toronto Metropolitan University\\
350 Victoria Street}
\email{dario.crisci@ryerson.ca}

\author{S.E.Ferrando}
\address[S.E.Ferrando]{Department of Mathematics\\
Toronto Metropolitan University\\
350 Victoria Street}
\email{ferrando@torontomu.ca}

\author{K.Gajewski}
\address[K. Gajewski]{Department of Mathematics\\
Toronto Metropolitan University\\
350 Victoria Street}
\email{konrad.gajewski@ryerson.ca}



\begin{abstract}
An agent-based  modelling methodology for the joint price evolution of two stocks is put forward. 
The method models future multidimensional price trajectories reflecting how a class of agents rebalance their portfolios in an operational way by reacting to how stocks' charts unfold.
Prices are expressed in units of a third stock that acts as numeraire.
The methodology is robust, in particular, it does not depend on any prior probability or analytical assumptions 
and it is based on constructing scenarios/trajectories. A main ingredient is a superhedging interpretation that provides relative superhedging
prices between the two modelled stocks. The operational nature of the methodology gives
objective conditions for the validity of the model and so implies realistic risk-rewards
profiles for the agent's operations. Superhedging computations are performed with
a dynamic programming algorithm deployed on a graph data structure. 
Null subsets of the trajectory space are directly related to arbitrage opportunities (i.e. there is no need for probabilistic considerations) that may emerge during the trajectory set construction. It follows that the superhedging algorithm handles null sets in a rigorous and intuitive way. Superhedging and underhedging bounds are kept relevant to the investor by means of a worst case pruning method and, as an alternative, a theory supported pruning that relies on a new notion of small arbitrage.
\end{abstract}


\maketitle
\section{Introduction}

A class of non-probabilistic models is proposed for the joint evolution of  prices for two assets   in units of a third asset that acts 
as a numeraire. The models construct future price scenarios (also referred to as trajectories).  One asset is singled out, arbitrarily, as a target asset to be superhedged by trading on the two remaining assets (one of them the numeraire asset). In other words, the constructed portfolio value superhedges the target asset providing a model perspective for relative pricing among the assets. The models are constructed according to a general methodology
centered on an agent-based operational framework. More precisely, a class of agents is implicitly defined by how they react to specific observed price changes through portfolio rebalances. Such an operational point of view provides an agent-based perspective to the obtained  price bounds of the target asset in terms of the other two. Relative price bounds are given by a superhedging portfolio and a dual portfolio for underhedging; the bounds can be used to asses a profit and loss profile relative to the agent's operations.
The trajectorial joint price model construction is effected by a general method 
based on observable quantities and
is defined through a set of combinatorial possibilities restricted through worst case historical constraints. Arbitrage opportunities may arise during the process of constructing scenarios; they are later  treated as null sets, a notion defined independently of probabilistic considerations and financially motivated.

The paper illustrates how a recently established non-probabilistic framework to model asset prices (see \cite{bender1, bender3, degano},\\ \cite{degano2, ferrando2, ferrando}) can be used
in a practical trading setting. To explain our contributions we first introduce, informally, the main ingredients of the theory (which are introduced precisely elsewhere in the paper) in terms of the constructions and notation of the present paper. The basic object of the theory is a {\it trajectory set} $\mathcal{X}$, ${\bf X} \in \mathcal{X}$ are sequences ${ \bf X} = \{ {\bf X}_i= (X_i, Z_i)\}_{i \geq 0}$
and $X_i= (X_i^1, X_i^2, \ldots, X_i^d)$ represent future modelling prices of assets $X^k$. The {\it additional coordinates} $Z_i$ (e.g. time of rebalance, quadratic variation, etc) are only used during the trajectory set  construction and are not directly involved in pricing considerations. The set $\bf {X} $ is decomposed in {\it conditional trajectory sets} $\mathcal{X}_{({\bf X}, i)}$, these are subsets of the unconditional trajectory set $\mathcal{X}$  and satisfy $\hat{\bf X} \in \mathcal{X}_{({\bf X}, j)}$ if $\hat{\bf X}_k = {\bf X}_k, 0 \leq k \leq j$.  There are no topological or measure theoretic properties required of the set $\mathcal{X}$; moreover, there are no cardinality restrictions, in particular, there could be an uncountable number of trajectories.
The conditional sets allow to embody $\mathcal{X}$ with several notions of arbitrage, the latter play a key role to define null sets and to support the constructions of the main analytical objects. Namely, 
superheding functionals $\overline{\sigma}_jf({\bf X})$ representing the funds required to superhedge $f: \mathcal{X}_{({\bf X}, j)} \rightarrow \mathbb{R}$, i.e. $f(\tilde{\bf X}) \leq \overline{\sigma}_jf({\bf X}) + \sum_{i=j}^{N-1} H_i(\hat{\bf X}) \cdot \Delta_i \hat{X}$ where $H_i(\hat{\bf X}) \cdot \Delta_i \hat{X}$ is the inner product in $\mathbb{R}^d$ between the holdings $(H_i^1, \ldots, H_i^d)$ of a self-financing portfolio (non-anticipative) and the price changes in the model i.e. $\Delta_i \hat{X}= (\hat{X}^1_{i+1}- \hat{X}^1_{i}, \ldots, \hat{X}^d_{i+1}- \hat{X}^d_{i})$. The models in this paper are typically incomplete as they are built using historical increments and allowing for possible combinations for the future. This construction is agent dependent as historical events are filtered accordingly on how the investor re-adjusts her portfolio.
Justifications for such an operational approach are developed in \cite{ferrando19}. $\overline{\sigma}_jf({\bf X})$  carries model information on (superhedging) price relationships among future payoffs $f$. This price information is relevant to the agent as the model is constructed in a way that reflects her rebalancing operations.  Here are the main contributions of the paper:
\begin{enumerate}
\item We provide algorithmic details to construct a specific set $\mathcal{X}$ in a practical setting where the agent trades during a day using three assets. The main ingredients in this scenario construction represent general ideas
that, we expect,  could be translated to other financial settings. In particular, we show several novel ways to trade risk and reward. The notion of risk refers, in our setting, to uncertainty on a possible set of scenarios/trajectories i.e. it is a non-probabilistic notion. 

\item We present output indicating how an agent could invest in a particular trading day depending on the initial conditions and how much risk is she willing to take by modifying the trajectory set. In other words, including or removing specific trajectories in the model are objective operations given that these trajectory operations are related to market observations. This should be contrasted to a purely data driven model construction where
it is usually a problem to interpret characteristics of the model (e.g. why a particular modelling path has been or has not been included).

\item Advantages of the model construction are highlighted; for example the calibration is empirically driven
and essentially depends on a single parameter (more generally, it depends of the agent's operations). This fact should be contrasted to the perilous task of calibrating stochastic processes to market data. 

\item We present two general ways to prune trajectories in order to trade risk and reward. The first method relies on pruning trajectories by modifying worst case calibrated parameters (or, relatedly, by the shrinking of an associated convex hull). The risk-reward repercussions of such pruning are objective as one can identify the empirical reasoning for disposing with the prunned trajectories. The other pruning method is theory based and depends on a new notion of {\it small arbitrage}, it is illustrated by an example that relies on a non-probabilistic extension of Dubin's classical upper bound for upcrossings.

\item Up to the best of our knowledge, our paper is first to provide a detailed and general price modelling methodology based on a rigorous non-probabilistic framework (we continue and extend our previous work in \cite{ferrando19}). Our reliance on theory is substantial, not only the framework provides guidance but it offers a natural and practical way to handle (non-probabilistic) null sets and to introduce the new notion of small arbitrage.

\item We also provide several practical discussions and illustrations, e.g. we show how trajectory matching can be naturally performed in our model. Such trajectory matching is crucial for pathway hedging and it is
uncommon to see it reported in the literature. Profit and loss analysis is provided illustrating that the methodology is natural and of a practical character.

\item The theory handles arbitrage opportunities, which may potentially appear during the trajectory set construction process, as null events in a rigorous and intuitive way that can be implemented precisely in a computer.
We take advantage of this feature in our superhedging algorithm and trajectory construction. 
\end{enumerate}

Our paper provides a concrete approach for constructing trajectory spaces and illustrates the financial implications of this framework. We treat null sets rigorously, tying them to the constructive properties of the trajectory space. The main focus is on evaluating superhedging and underhedging functionals, with the price interval generated by these methods providing bounds for asset prices. A key concern is whether this interval is informative. Our approach avoids the problem of wide intervals by controlling worst-case scenarios, ensuring more relevant pricing information.  


Once superhedging prices are computed, an agent can evaluate the profit and loss (P\&L) profile using the trajectory space model. We also explore how risk can be traded for reward by shrinking the trajectory set. This objective process prunes extreme events, thereby narrowing the price interval. By adjusting worst-case parameters, we could incorporate observed historical frequencies to account for probabilistic risk in investment decisions.


\subsection{Relation to the Literature}
Modern mathematical finance relies on the principle of no-arbitrage, which states that trading strategies should not allow to make a risk-free profit. In essence, models should ensure that any potential market outcome involves some possibility of loss. No-arbitrage is typically enforced by associating arbitrage opportunities with (probabilistic) null events in the model.

No-arbitrage dictates the price of financial positions in complete markets relative to underlying asset prices. This price reflects the initial investment needed to create a perfect hedging portfolio, though real-world situations often require deviations from idealized assumptions (e.g., continuous time, infinite transactions). In incomplete markets, no-arbitrage restricts the price dynamics in a stochastic setting by enforcing that the discounted price process is a martingale. This paper restricts its focus to discrete time.

The established (stochastic) mathematical formalism
critically depends on the theory of probability, in particular, 
for  the notions of null set and martingale. Recently, there has been a conceptual shift
to recast financial concepts, no-arbitrage included, in non-probabilistic terms. The resulting 
modelling approaches are labelled as {\it robust} (modelling) indicating that fewer assumptions are
required to set up models for the  time evolution for assets prices. 
To the best of our knowledge, theoretical developments in this area, such as those by \cite{liebrich}, \cite{bartl}, \cite{blanchard} and \cite{burzoni},  have not yet led to practical robust model proposals. 
An obstruction to develop practical models that are supported by current theoretical approaches is the complex nature of null events, they seem to require
mathematical machinery that will complicate a computer implementation
and make the relevance of the theory, inasmuch as null sets are concerned, in an actual model, ambiguous or inconsequential.

The present work extends the operational methodology originally presented in \\
\cite{ferrando19}, it now allows the superhedging of any assets (as opposed to superhedging only options). This allows to make the methods much more easily deployable, e.g. we provide output for daily traiding. Moreover, the models we develop are more stylized and the calibration procedure is
streamlined. Most substantial is the fact the framework of the present paper relies on the theory developed in \cite{bender3} which, in particular,  allows us to deal with null sets in a rigorous way. To trace back the origins of our framework, we mention the  papers \cite{ferrando}, \cite{degano}, \cite{degano2} that lead to a
robust methodology based on the notion of a trajectory space. The latter is a structured version of an abstract  probability space and has the capability to encode several no-arbitrage related notions. We refer to such theory and setting with the word {\it trajectorial}. We take advantage of such robust framework, in particular, the proposed trajectories are built to reflect objective features that affect our class of investors. Such flexibility is available as the manifold of possible trajectories is not constrained by probabilistic or analytical assumptions. The main body of the paper introduces a minimal amount of theory as that is not our main emphasis, a summary of the theory and the results we depend upon are presented in an appendix.

The paper is organized as follows, Section \ref{sec:setting} briefly describes the formal mathematical setting. Section \ref{a.e. section} introduces the conditional norm operator, leading to null sets, and the conditional superhedging operator. Section \ref{sec:computation} explains how one transitions from the theory to actual computations of superhedging prices. A more detailed explanation is provided in Appendix \ref{theoreticalFramework}. The extensive Section \ref{sec:trajectorySetConstruction} describes the trajectory set construction, the operational use of general data, the general definitions of Models A and B, the definition of the empirical set $N_E$ used to iterate the trajectory set construction and a first introduction to pruning constraints (the latter topic is completed in Appendix \ref{additional_pruning_constraints}). Section \ref{modelSpecification} fully specifies the models used in the paper as well as the method of Dynamic Pruning. Section \ref{dataAndCalibration} introduces specific data, a calibration method and several output displays. Geometric Brownian motion simulated data is used to check some basic feautures of the methodology. The section also provides some trajectory matching output.
Section \ref{sec:profitAndLossAnalysis} provides output and explanations for a profit and loss analysis that we perform to illustrate the methodology. Section \ref{sec:Arbitrage} explains how arbitrage opportunities could apppear during trajectory set construction and how
our theory allows to handle those situations. The section also introduces a new notion of {\it small arbitrage} and describes a novel $2$-dimensional non-probabilistic  setting for Dubin's bound for upcrossings. We explain how such type of result can be used for a different, apriori theory-based, type of trajectory pruning.  Section \ref{sec:discussion}  presents a discussion that allows to see the proposed approach from a general point of view. Appendix \ref{theoreticalFramework} introduces basic definitions from our theoretical framework and describes in detail how theoretical superhedging quantities are calculated by an algorithm. Appendix \ref{additional_pruning_constraints} completes the description of the pruning constraints used to construct our trajectory models.

\section{Trajectorial Setting} \label{sec:setting}

We briefly introduce the mathematical setting and provide references to detailed developments
and proofs (additional background is provided in Appendix \ref{theoreticalFramework}). Our models propose prices for a finite number of assets whose initial values are known and evolve in discrete time. Conditioning
on given past information,  uncertainty is prescribed by the fact 
that potential future prices belong  to a set of multidimensional sequences, the latter we call trajectories and play the role of possible scenarios. Therefore, the models are non-deterministic in a non-probabilistic way.
The trading strategies are given by portfolios that will be successively re-adjusted, taking into account the information available at each stage. Some justification for the setting can be found in \cite{ferrando} (one dimensional case) while the paper \cite{degano2} provides details on the multidimensional case. Notice that \cite{degano2}  defines the superhedging operator by means of a single portfolio and in finite time while
references \cite{bender1}, \cite{bender3} and \cite{ferrando2} present the general theory
allowing for an infinite number of superhedging portfolios and infinite time (but the setting is one-dimensional). Here we will rely on a multidimensional setting as well as on the
availability of an infinite number of portfolios and so we will combine the two streams of literature listed above. We prove, see Theorem \ref{thm:sigma_is_correct}, that the infinite number of portfolios are only used to define null events in trading terms. More specifically, the superhedging bounds with an infinite number of portfolios can be replaced with a single simple portfolio if we allow the upperbound to hold a.e. (this latter notion is not probabilistic and introduced at due point in our paper). The use of a countable
number of portfolios is completely analogous to
the case of Lebesgue's measure.

The paper
describes a systematic method to construct a particular class of trajectory sets 
and so providing an example of the general framework
that we describe next.
We consider a financial market with $d+1$ assets that evolve in a fixed time interval $[0,T]$ (the case $[0, \infty)$ is also allowed in much of the developments). At some point, we will require $d=2$ for practical reasons (as explained at due time). Our models will be discrete in the sense that the trading instances are indexed by integer numbers. Given $s_0=(s_0^0,s_0^1,s_0^2,\dots,s_0^d) \in \mathbb{R}^{d+1}$, as initial prices of assets $S^k$,  potential future prices are modeled by sequences taking values in $\mathbb{R}^{d+1}$ with coordinates
\[ S_i=(S^0_i,S^1_i,S^2_i,\dots,S^d_i)~~ \mbox{with}~~ S_0=s_0. \]
In fact we will model discounted prices as we explain shortly.

\vspace{.1in}
To be definite, we will consider that the real numbers $S^k_i$ express the price of asset $S^k$ in a common currency, a unit of which we denote generically by $\$$. That is, in terms of dimensions $[S^k_i]= \$/{\bf 1}_{S^k}$ where ${\bf 1}_{S^k}$ is one unit of asset $S^k$.
 On the other hand, for financial reasons (see \cite{vecer}), it is important to work with an arbitrary reference asset; this is achieved by taking a reference asset as \emph{numeraire}. For example, in some cases it is useful to select the value of a bank account as numeraire.

The trajectories we will construct are multidimensional and elements of a trajectory set denoted by $\mathcal{X}$. Elements ${\bf X} \in \mathcal{X}$  are of the  form ${\bf X}= \{{\bf X}_i\}_{i\geq 0}$ where ${\bf X}_i =(X_i, \ldots)$, the $X_i$ are $d$-tuples with coordinates 
$X_i= (X^1_i, \ldots, X^d_i) \in \mathbb{R}^d$  that represent the prices of assets $S^k$ in units of asset $S^0$ respectively. Additional coordinates in ${\bf X}_i$, indicated by $\ldots$, will be described shortly.  Specifically, the units are 
$$
[X_i^k]= \frac{{\bf 1}_{S^0}}{{\bf 1}_{S^k}}, ~1 \leq k \leq d,
$$
where ${\bf 1}_{S^j}$ is one unit of asset $S^j$ and  the asset $S^0$ is the numeraire.
Clearly, the choice of numeraire is arbitrary as long as the condition $S^0_i \neq 0$, $i \geq 0$, is satisfied. In this paper, we refer to trajectories by ${\bf X}$, this is in contrast to related work
(e.g. \cite{bender1}  and \cite{ferrando2})  where $S$ is used. We are interested in modelling the former  i.e. the variables  discounted by an arbitrary numeraire. 
To relate to the $1$-dimensional trajectory sets introduced in the mentioned papers we may think that
we were taking $S_i = X^1_i$ and ${\bf 1}_{S^0}= {\bf 1}_B$ representing one unit of a bank account (i.e. we deal with ``discounted'' prices). 

The numerical value of $ X^k_i$ (i.e. stripped from its units), is the number of units of the asset $S^0$, now the numeraire, which are required to acquire one unit of the $S^k$ asset. As noted, we will model the sequences  ${\bf X}$ directly, i.e. without any reference to
the original values $S$ and, consequently, we will rely on the notation ${\bf X} = \{(X_i, \ldots) \} _ {i \ge 0}$.

\begin{definition}[$d$-dimensional Trajectory set]  \label{trajectories}
Consider  $\Sigma =\{\Sigma_i\}$  a given family of subsets of $\mathbb{R}^{d}$, 
$\Omega=\{ \Omega_i \}$ is a family of sets.
For given $x_0 \in \mathbb{R}^{d}$ and $z_0 \in \Omega_0$, a \emph{trajectory set} $\mathcal{X}$ is a subset of
\begin{equation}  \nonumber
\mathcal{X}_{\infty}(x_0, z_0) \equiv \left\lbrace {\bf X} = \{{\bf X}_i \equiv (X_i,Z_i)\}_{i \geq 0}: ~X_i \in \Sigma_i, ~ Z_i \in \Omega_i \right\rbrace,
\end{equation}
such that $(X_0, Z_0)= (x_0,z_0)$. The elements of $\mathcal{X}$  will be called \emph{trajectories}.
\end{definition}
The numbers $X_i= (X_i^1, \ldots, X_i^d)$ are called the {\it traded coordinates} while the
$Z_i$ are called {\it additional coordinates}.
We emphasize that for any conceivable situation of interest  $\mathcal{X} \neq \mathcal{X}_{\infty}(x_0,z_0)$ and
we have all the flexibility we may wish to design $\mathcal{X}$ for practical purposes.
Constructions of interesting sets $\mathcal{X}$ is a main concern of the paper.
The only theoretical restrictions imposed on $\mathcal{X}$ will be very general no-arbitrage 
constraints that we will incorporate at due time (see Definition \ref{propertyL} and the follow up notion of $(L)-a.e.$). The portfolio re-balancing stages may be triggered by arbitrary events of the market without the need to be directly associated with time. This greater degree of generality is handled by the additional coordinates $Z_i$ that add a new source of uncertainty to the trajectories' coordinates (these additional coordinates play an important role when constructing specific models). The $Z_i$ can be set-valued; in financial terms, this new variable can represent any collection of variables of interest such as volume of transactions, time, quadratic variation of trajectories, etc. (see for example \cite{ferrando19}). The additional coordinates are not part of the explicit portfolio composition, in particular, they are not being traded and are used for the purpose to construct, and then prune, the trajectory set. Once the latter is built, the additional coordinates can be neglected for superhedging and subsequent analysis.  For simplicity and whenever
the additional coordinates $Z_i$ play a secondary role in a discussion, we will suppress
such coordinates and refer to a trajectory simply by $X = \{X_i\}_{i \geq 0}$.
In fact, as it will be clear through our developments, the additional coordinates are only used when creating the trajectory set and are not used in the follow up step of computing superhedging quantities (nor are they needed in the definition of the superheding operators).

It is important to note that if $\tilde{{\bf X}} = \{(\tilde{X}_i, \tilde{Z}_i) \} $ and $ \hat{{\bf X}} = \{(\hat{X}_i, \hat{Z}_i) \} $ are two trajectories, $ \tilde{{\bf X}}_i $ could unfold at a different time than $ \hat{{\bf X}}_i $. That is, the index $i$ will be associated with portfolio re-balances stages but they will not be necessarily associated to (uniform) time. It is only required that the stage $i + 1$ occurs temporarily after the stage $ i $ for each trajectory.


At the $ k$-th stage, the information about the future available to investors is that ${\bf X}$ is an element of the set
\begin{equation} \nonumber
\mathcal{X}_{({\bf X},k)}\equiv \left\lbrace {\bf X}' \in \mathcal{X}: {\bf X}'_i= {\bf X}_i, 0 \le i \le k \right\rbrace \subseteq \mathcal{X}, ~\mbox{and so}~~\mathcal{X}= \mathcal{X}_{({\bf X},0)}.
\end{equation}
The notation $({\bf X},k)$ will be referred to as  \emph{node} and acts as a shorthand notation for the set $\mathcal{X}_{({\bf X},k)}$ called the \emph{trajectory set conditioned at the node $({\bf X},k)$}. The future information contained in $ \tilde{{\bf X}} \in \mathcal{X}_{({\bf X}, k)} $ depends on the past only through ${\bf X}_0, \ldots, {\bf X}_k$.
The multiple number of trajectories  emanating from a node
reflects the non-deterministic nature of the assets' model time evolution.
As trajectories unfold more coordinates become available and so the investor increases his knowledge about possible future scenarios. This is expressed by the fact
\[ \mathcal{X}_{({\bf X},k')} \subseteq \mathcal{X}_{({\bf X},k)},\]
for $k'>k$. The following notation will also be used
\begin{equation} \label{eqn:conjdelta}
  \Delta X(\mathcal{X}_{({\bf X},k)}) \equiv \{ \Delta_k \tilde{X}:~ {\bf \tilde{X}} \in \mathcal{X}_{({\bf X},k)} \} \subseteq \mathbb{R}^{d},
\end{equation}
where $\Delta_k \tilde{X}= \tilde{X}_{k+1}- \tilde{X}_k= (\tilde{X}^1_{k+1}, \ldots, \tilde{X}^d_{k+1}) - (\tilde{X}^1_{k}, \ldots, \tilde{X}^d_{k})$ (i.e we take the differences of only the traded coordinates) . We will refer to any property as {\it local} if it is relative to a node $ ({\bf X}, k) $ and only involves elements of $\Delta X(\mathcal{X}_{({\bf X},k)})$.

\subsection{Conditional Portfolio Sets}  \label{portfolioSets}

The other basic component are \emph{portfolios} defined as follows.

\begin{definition}[Conditional portfolio set] {\cite[Definition 2]{bender1}} For any fixed $ {\bf X} \in \mathcal{X} $ and $j\ge 0$, $\He_{({\bf X},j)}$ will be a set of sequences of functions $H = \{{\bf H}_i = (H_i^0, H_i)\}_{i \geq j}$, where $H_i: \mathcal{X}_{({\bf X},j)} \rightarrow \mathbb{R}^{d}$ and $H^0_i: \mathcal{X}_{({\bf X},j)} \rightarrow \mathbb{R}$ are non-anticipative in the following sense: for all $\tilde{{\bf X}},\hat{{\bf X}}\in \mathcal{X}_{({\bf X},j)}$ such that $\tilde{{\bf X}}_k = \hat{{\bf X}}_k$ for $j \le k\le i$, then ${\bf H}_i(\tilde{{\bf X}}) = {\bf H}_i(\hat{{\bf X}})$ (i.e. ${\bf H}_i(\tilde{{\bf X}}) = {\bf H}_i(\tilde{{\bf X}}_0,\ldots, \tilde{{\bf X}}_i)$). We will assume $\He_{({\bf X},j)}$ to be a vector space
for any possible pair $({\bf X}, j)$. Furthermore, portfolios will be sel-financing as in Definition \ref{def:selfFinancing} below.
\end{definition}
\noindent
Notice that $[H^k_i]={\bf 1}_{S^k}$. We will write $H_i({\bf X}) \cdot Y$, $Y \in \mathbb{R}^d$,  for the Euclidean inner product in $\mathbb{R}^d$.
In general, $\He_{({\bf X},j)}$ does not need to include all possible sequences of available non-anticipative functions and we refer to \cite{bender1} for details on this particular issue.
$H \in \He_{({\bf X},j)}$ may be referred to as a \emph{ conditional portfolio}.

\vspace{.1in}
It will also be convenient to define \emph{global portfolios}. For a fixed $j \geq 0$,  $\He_j$ denotes the set of sequences of functions $H\equiv \{{\bf H}_i = (H_i^0, H_i)\}_{i\ge j}$ with $H_i: \mathcal{X} \rightarrow \mathbb{R}^d$
and  $H^0_i: \mathcal{X} \rightarrow \mathbb{R}$ where for each $ {\bf X} \in \Se$ there exists $G \in \mathcal{H}_{({\bf X},j)}$ such that ${\bf H}_i(\tilde{{\bf X}})= {\bf G}_i(\tilde{{\bf X}}) \;\;\forall~ \tilde{{\bf X}}\in \Se_{({\bf X},j)}~~\mbox{and}~~i \geq j$. A global portfolio $H$  could be characterized by indicating that its restriction to $\Se_{({\bf X},j)}$ belongs to $\He_{({\bf X},j)}$.

$ H^k_i (\se) $ represents the number of units held for the $k$-th asset during the period between $ i $ and $ i + 1 $. Therefore, $ H_i^0({\bf X}) + H_i({\bf X}) \cdot X_i $ is the {\it value}, in units of the numeraire, of the assets' holdings $(H_i^0, H_i)$  at stage $ i $, while $ H_i^0({\bf X}) + H_i({\bf X}) \cdot X_{i+1} $
 is the value  just before rebalancing at the end of the period. 

In the next re-balancing, the investor will invest ${\bf H}_{i + 1}$; in general, $ H^0_{i + 1}({\bf X}) + H_{i + 1}({\bf X}) \cdot X_{i + 1} $ may be different from $ H^0_{i}({\bf X}) + H_{i}({\bf X}) \cdot X_{i + 1} $. In this latter case, it follows that some units of the assets were added or removed, without replacement, from the portfolio. However, this situation is precluded for many applications. For example, if the goal is to look for a ``fair'' price for a certain financial contract, this value should be the minimum necessary to cover the obligations generated by the contract, that is, any injection or withdrawal of money will affect this property. This reasoning justifies the use of the following concept.

\begin{definition}[Self-financing portfolio] \label{def:selfFinancing}
A conditional portfolio $H$ is called \emph{self-financing} if for all $\se' \in \mathcal{X}_{({\bf X},j)}$ and $i \geq j$,
\begin{equation} \label{selfFinancing}
H^0_{i}({\bf X}')~+ H_{i}({\bf X}') \cdot X'_{i+1} = H^0_{i+1}({\bf X}')~+ H_{i+1}({\bf X}') \cdot X_{i+1}'.
\end{equation}
\end{definition}

The self-financing property means that the portfolio is re-balanced in such a way that its value is preserved. From this property it is clear that the accumulated gains and losses resulting from price fluctuations are the only sources of variation of the portfolio; then for ${\bf X}' \in \mathcal{X}_{({\bf X}, j)}$
\begin{equation} \nonumber 
H^0_{k}({\bf X}')~+ H_{k}({\bf X}') \cdot X_{k}' = H^0_{j}({\bf X})+ H_{j}({\bf X}) \cdot X_{j} + \sum_{i=j}^{k-1}  H_{i}({\bf X}') \cdot \Delta_i X', ~k \geq j, 
\end{equation}
$\mbox{where}\;\: \Delta_i X' = X'_{i+1}- X'_i.$


\vspace{.1in}
\noindent
For a fixed node $({\bf X},j)$, $H\in\He_{({\bf X},j)}$, $V\in \mathbb{R}$ and $n\ge j$ we define $\Pi_{j,n}^{V, H}: \mathcal{X}_{({\bf X},j)} \rightarrow \mathbb{R}$, for $\tilde{{\bf X}} \in \mathcal{X}_{({\bf X},j)}$ by:
\begin{equation}   \nonumber 
\Pi_{j,n}^{V, H}(\tilde{{\bf X}}) \equiv V +\sum_{i=j}^{n-1}H_i(\tilde{{\bf X}}) \cdot \Delta_i \tilde{X},~ n \geq j,
\end{equation}
and $H^0_j({\bf X})= V - H_j({\bf X}) \cdot X_j$.
Notice that as $({\bf X},j)$ changes, the given  $V$ could change as well and so in effect
$V= V({\bf X}_0, \ldots, {\bf X}_j)$.

\vspace{.1in}
As an abuse of language, the notation $X^k$, besides being used for the prices $X^k_i$, will also be used to refer to asset $k$. In the sequel, being $\mathcal{A}$ a set of real valued functions, $\mathcal{A}^+$ will denote the set of its non-negative elements.

\begin{definition} [Elementary vector spaces]  
\vspace{.05in} For a fixed node $({\bf X},j)$  set
\begin{equation} \label{conditionalElementarySpace0}
\mathcal{E}_{({\bf X},j)}= \{f = \Pi_{j, n_f}^{V, H}: H \in \He_{({\bf X},j)},~~V \in \mathbb{R}~~~\mbox{and}~~~n_f \in \mathbb{N}\}.
\end{equation}
Also, $\mathcal{E}_{({\bf X},j)}^+$ will denote the non-negative elements of $\mathcal{E}_{({\bf X},j)}$. The sets in (\ref{conditionalElementarySpace0}) are vector spaces given our assumption that each $\He_{({\bf X},j)}$ are themselves vector spaces. Elements of $\mathcal{E}_{({\bf X},j)}$ are called {\it elementary functions}. Let us also define
\begin{equation} \nonumber 
\mathcal{E}_j = \{f:\Se\rightarrow \mathbb{R}: f\vert_{\Se_{({\bf X},j)}}\in\mathcal{E}_{({\bf X},j)} \;\; \forall {\bf X} \in \Se \}.
\end{equation}
\end{definition}

\section{Fundamental Operators and Almost Everywhere Notion}\label{a.e. section}
Let $Q$ denote the set of all functions from $\mathcal{X}$ to $[-\infty, \infty]$ and $P \subseteq Q$ denotes the set of all non-negative functions. The following conventions are in effect: $0 ~\infty =0$, $\infty + (- \infty) = \infty$, $u - v \equiv u +(-v) ~\forall~u, v \in [-\infty, \infty]$, and $\inf \emptyset = \infty$. $f \in Q$ is said to be of {\it finite maturity}
if $f({\bf X})= f({\bf X}_0, \ldots, {\bf X}_n)$ for some constant $n$ ($n$ is called the maturity time of $f$).

We define next the \emph{conditional norm operator} \; $\overline{I}_j : P \rightarrow \mathcal{E}^+_j$, it is used to define null sets by taking $j=0$. 

\begin{definition}\label{Iup_definition}
For a given node $({\bf X}, j)$ and a general $f \in P$, define:
\begin{equation} \nonumber
\overline{I}_j f ({\bf X})\equiv  \inf \left\{\sum_{m \geq 1} V^m: ~~f \leq  \sum_{m \geq 1} ~~\Pi_{j, n_m}^{V^m, H^m}\;\;\;\mbox{on}\;\; \mathcal{X}_{({\bf X},j)}\;
\right\},
\end{equation}
where $\Pi_{j, n}^{V^m, H^m}\in \mathcal{E}_{({\bf X},j)}^+  ~~\forall ~~j \leq n \leq n_m$.
\end{definition}

\vspace{-.1in}
The requirement $\Pi_{j, n}^{V^m, H^m}\in \mathcal{E}_{({\bf X},j)}^+$, for all $n,~ j \leq n \leq n_m$ (as contrasted 
to the single case $n= n_m$) is only needed to handle some arguments related to the case of type II nodes (types of nodes are introduced in Section \ref{nodeTypes}). We will use the notation $\overline{I}f \equiv \overline{I}_0f$. 

\noindent
We also set, for a general $f \in Q$:
\begin{equation} \nonumber
\left \Vert f \right \Vert _j({\bf X}) \equiv \overline{I}_j \vert f \vert({\bf X})~~\mbox{and}~~\left \Vert f \right \Vert \equiv \left \Vert f \right \Vert_0({\bf X}).
\end{equation}

\noindent
Notice that $\overline{I}_j f ({\bf X})= \overline{I}_j f ({\bf X}_0, \ldots, {\bf X}_j)$, i.e. $\overline{I}_j f (\cdot)$ is constant on $\mathcal{X}_{({\bf X},j)}$. 
Moreover $\overline{I}_j f\ge 0$, so $\Vert 0 \Vert_j=0$.
$\Vert \cdot \Vert_j({\bf X})$ will be called a {\it conditional norm}.

\vspace{.1in}
Next we introduce the notions of \emph{conditional null set} and the \emph{conditional $a.e.$ property}.
\begin{definition}[Conditional a.e. notions] \label{nullObjects} Given a node $({\bf X},j)$, a function $g\in Q$ is a \emph{conditionally null function at} $({\bf X},j)$  if:\[\|g\|_j({\bf X})=0.\]
 A subset $E\subset\Se$ is a \emph{conditionally null set at $({\bf X},j)$} if $\|\mathbf{1}_E\Vert_j({\bf X})=0$. A property is said to hold conditionally a.e. at $({\bf X},j)$
(or equivalently: the property holds ``a.e. on $\mathcal{X}_{({\bf X},j)}$") if the subset of  $\Se_{(S,j)}$ where it does not hold
is a conditionally null set at $({\bf X},j)$. In particular, the latter definition applies to $g=f$ a.e. on $\mathcal{X}_{({\bf X},j)}$. Unconditional null sets correspond to the case $j=0$ and are also called global null sets, this case will be simply denoted $a.e.$ and it applies to equalities as well as inequalities of course.
\end{definition}

\vspace{-.1in}
All appearing equalities and inequalities are valid for all points in the spaces where the functions are defined unless qualified by an explicit a.e. 

\vspace{.1in}
We introduce next the operator  $\overline{\sigma}_j : Q \rightarrow \mathcal{E}_j$,
which we will call a \emph{conditional superhedging operator} (also called {\it conditional outer integral} in \cite{bender1}); it is a crucial notion in our setting.

\begin{definition}\label{cond_integ_def} For a node $({\bf X},j)$ and a general $f \in Q$,
\begin{equation} \nonumber
\overline{\sigma}_j f({\bf X}) \equiv  \inf \left\{\sum_{m \geq 0}V^m: ~~f \leq  \sum_{m \geq 0} f_m \;\;\mbox{on}\; \Se_{({\bf X},j)}\right\},
\end{equation}
where
$f_0=\Pi_{j, n_0}^{V^0,H^0}\!\!\in \mathcal{E}_{({\bf X},j)} \; ~~\mbox{and, for}~ m \geq 1, 
\;\;f_m \equiv \Pi_{j, n}^{V^m, H^m}\!\!\in \mathcal{E}_{({\bf X},j)}^+~~\forall ~~~n \geq j.$
Define also $\underline{\sigma}_j f({\bf X}) \equiv -\overline{\sigma}_j(-f) ({\bf X})$.
We will use the notation $\overline{\sigma}f \equiv \overline{\sigma}_0 f$.
Note that $\overline{\sigma}_j f({\bf X})= \overline{\sigma}_j f({\bf X}_0, \ldots, {\bf X}_j)$. 
\end{definition}

\subsection{Null Sets as Unlikely Financial Events}    \label{detectingNullEvents}

\noindent
Here we indicate the intuition behind the definition of the operator
$\overline{\sigma}_jf$ (with analogous explanations for $\overline{I}_jf$). The main
simple portfolio superhedging $f$ is given by $f_0+ \sum_{m=1}^N f_m$ for $N$ sufficiently large, the role of the idealization
of an infinite number of  non-negative  portfolios
$\sum_{m \geq 1} f_m$ is used to detect arbitrage nodes (see Appendix \ref{nodeTypes} 
for the definitions of types of nodes) and so to define null sets.
This is in close analogy to the use of elementary regions in order to define the area
of  non-elementary regions of the plane by means of Carath$\acute{\mbox{e}}$odory's outer measure. In our case,  elementary regions are replaced by simple portfolios, a class closed under linear combinations and playing the role of the simple functions in Lebesgue's theory of integration. The quantity $\overline{\sigma}_jf$
is then interpreted as a conditional outer integral (we are conditioning on $({\bf X},j)$ and hence restricting
the future to $\mathcal{X}_{({\bf X},j)}$), a functional version of
Carath$\acute{\mbox{e}}$odory's outer measure (see \cite{VdVW}).

Notice the inclusion of a single
simple portfolio $f_0$ with arbitrary sign in Definition \ref{cond_integ_def}; idealized portfolios of the
form $\sum_{m \geq 1} f_m$, $f_m \geq 0$, are only used to detect/define null functions as we argue below. The definition of null function has a purely trading nature: $f$ is null if for any given $\epsilon \geq 0$
we will have $|f| \leq \sum_{m \geq 1} f_m$ with $\sum_{m \geq 1}V^m \leq \epsilon$.
Take $f = c~{\bf 1}_A$ with $A \subseteq \mathcal{X}$ and $c >0$ for example, then, it costs arbitrarily little to superhedge $f$ while its payoff $c {\bf 1}_A$ could be arbitrarily high (relative to the investment $\epsilon$). This is analogous to a lottery where you pay arbitrarily little to get into the draw with the possibility of an arbitrarily large payment. That is, financial positions are considered null if they are cheap to superhedge by means of idealized portfolios allowing an unbounded number of transactions.
These events are then considered unlikely and in the theory are null sets. Arbitrage opportunities
will be null sets, for example suppose ${\bf X}_0, \ldots, {\bf X}_n$ has unfolded
and that $h \cdot (\hat{X}_{n+1} - X_n) >0$, for a given $h \in \mathbb{R}^d$ and  for all $\hat{\bf{X}} \in \mathcal{X}_{({\bf X},n)}$.
Such future represents an arbitrage opportunity by transacting with the numeraire and the asset.
We can then design (arbitrage) portfolios $f_m$ that are switched on at $\mathcal{X}_{({\bf X},n)}$ and giving
$\sum_{m \geq 1} f_m(\hat{{\bf X}}) = \infty$ for all $\hat{\bf{X}} \in \mathcal{X}_{(\bf{X},n)}$,
this allows to take $\epsilon =0$ as initial investment proving that 
$f = {\bf 1}_A= {\bf 1}_{\mathcal{X}_{({\bf X},n)}}$ is a null function.
The need to have an infinite number of terms in $\sum_{m \geq 1} f_m$
occurs because the realized gain determined by $H_n^m({\bf X}) \cdot (\hat{X}_{n+1} - X_n)$ can be arbitrarily small. Relying
on a countable collection of positive portfolios $f_m$, $m \geq 1$, leads to the modern theory of integration in that one can handle countable unions of null sets.

\section{Computation of Superhedging Prices} \label{sec:computation}

As we have argued above, it is then essential to rely on
 an infinite number of portfolios in the definition of $\overline{I}$ in order to handle countable collections of null sets as it is customarily done in Lebesgue's theory of integration. Theorem \ref{aEComputationWithSimplePortfolios} below shows that, in the computation of the superhedging functional, one can replace the infinite sum of elementary portfolios by a single elementary portfolio in the case of superhedging  a finite maturity function $f$ with the caveat that the superhedging inequality only holds a.e. Therefore, the idealization of using a countable number of portfolios allows to incorporate null sets in a natural way, namely, superhedging is achieved with a single simple portfolio and resulting on the same bounds with the understanding that superhedging may not hold on null sets (the result is proven for functions depending on a finite number of coordinates). This result is analogous to the classical, measure based, theory as presented in \cite{VdVW}.
The result then allows to include arbitrage nodes in the model as they lead to null sets
in such a way that the model's price bounds are not affected. This observation, namely that arbitrage nodes can be neglected when computing superhedging prices, is being used 
during our trajectory construction process which, as it will be detailed at due time, can introduce arbitrage opportunities in a natural way and so the latter can be avoided during the computation of the price bounds. 

\vspace{.1in}
Theorem \ref{aEComputationWithSimplePortfolios} below is (essentially) Theorem 6.1 from \cite{bender3}. Appendix \ref{theoreticalFramework}
provides the necessary concepts to make sense of the statement, in fact, that appendix presents the statement and proof of Theorem \ref{thm:sigma_is_correct} a slight re-statement of Theorem \ref{aEComputationWithSimplePortfolios}  which is more useful for our purposes. The precise meaning of the key hypothesis $(L)-a.e.$, required in Theorem \ref{aEComputationWithSimplePortfolios}, is provided 
in Appendix \ref{theoreticalFramework} and sufficient conditions for its validity are in \cite{bender3}.

For simplicity, the next result is stated  globally (i.e. at node $({\bf X},0)$) and for the $d=1$ case. The result appears in \cite{bender3} and it is available
for the one-dimensional case. In \cite{bender3} the converse statement is also established but we will not need that part of the result.
\vspace{-.1in}
\begin{theorem} \label{aEComputationWithSimplePortfolios}
Suppose that $(L)$-a.e. holds, $d=1$ and that $f: \mathcal{X} \rightarrow \mathbb{R}$ and bounded has maturity $n_f$, i.e. $f({\bf X})= f({\bf X}_0, \ldots, {\bf X}_{n_f})$ for every ${\bf X} \in \mathcal{X}$. Then
\begin{equation} \nonumber 
\overline{\sigma}f = \inf \{V: f \leq \Pi^{V, H}_{0, n}~~a.e.~\mbox{where}~~~\Pi^{V, H}_{0, n_f} \in \mathcal{E}_0\}.
\end{equation}
\end{theorem}

\vspace{.1in}
As mentioned, we will construct models for the time evolution of the traded coordinates 
$(X^1_i, X^2_i)$ i.e. a $d=2$ dimensional model. On the other hand,
superhedging will involve a $1$-dimensional portfolio as we only intend to find out how to trade
with asset $X^1$ at prices $X^1_i$  in order to superhedge asset $X^2$ (at a future specified time). The second component of the portfolio, namely the number of numeraire units of a third asset $S^0$, does not need to be an explicit part of the superhedging computation as it is a by-product of the self-financing constraint. More precisely, once our proposed algorithm  determines the amount 
$H_i({\bf X})$ required to superhedge, one then
evaluates $H^0_i({\bf X})$ by means of (\ref{selfFinancing}). Given
this computational setup, we could envision extracting $1$-dimensional trajectories $X^1= \{X^1_i\}_{i\geq 0}$ from multidimensional ones $X= \{X_i\}_{i\geq 0}=\{(X_i^1, X_i^2)\}_{i\geq 0}$, in this way one creates a $1$-dimensional trajectory set from a multidimensional one. 

We observe the following (types of nodes are introduced in Appendix \ref{nodeTypes}): a $2$-dimensional arbitrage-free node will be a $1$-dimensional
arbitrage-free node in either of the dimensions $X^1$ or $X^2$ (however  
a $2$-dimensional arbitrage node may not lead to a $1$-dimensional arbitrage node).
Intuitively, an arbitrage opportunity at a $2$-dimensional node may only be available if we can trade with both assets. On the other hand the possibility of losing money with a $2$-dimensional portfolio at a $2$-dimensional arbitrage-free node will still remain
if we trade only with a single coordinate. To sum up,
when we are calculating superhedging prices using a $1$-dimensional portfolio, we will be implicitly extracting 
$1$-dimensional trajectories (either for $X^1$ or $X^2$) from our $2$-dimensional construction and potentially face nodes that are of type II arbitrage nodes (in a $1$-dimensional sense). We use the fact (from \cite{bender3}) that trajectories passing through a type II node form a null set, and, as explained previously, neglecting arbitrage nodes will not affect the price.  We formalize this interplay between the $2$-dimensional construction and the $1$-dimensional trajectory extraction for further computation  in
the Appendix \ref{sec:Pricing}. In that appendix we also present the details of the superhedging
dynamical programming algorithm as well as how this procedure handles
arbitrage nodes.

\section{Trajectory Set Construction. Main Ingredients}  \label{sec:trajectorySetConstruction}

\vspace{.1in}
We place ourselves in the following financial context: we model the future joint evolution
of discounted  prices, in units of a third stock acting as numeraire, of two assets (we will deal with stocks). 
Our models address the following question:
how much does it cost to superhedge one asset using a second asset?
More specifically, the models allow to evaluate a relative superhedging amount, this is the investment required to set up a portfolio
that trades with one stock and the numeraire in order to superhedge the value (in numeraire units) of one share of the remaining asset. The symmetry of the methodology allows for the flexibility to swap the mentioned assets as well as considering underhedging or changing the numeraire. The modelling framework is robust and so, one expects that the model hypothesis will hold for future unfolding chart values. In other words, the methodology does not rely
on a probability measure which, one expects, may be necessary
to reduce the effect of extreme trajectories affecting prices (we come back to this key topic shortly). To perform the required
superhedging evaluations one needs to model the joint price evolution, this we do through a multidimensional trajectory space. The flexibility of the theoretical framework allows to include trajectories with realistic characteristics. In this paper, the construction of trajectories is dictated by the investment impact on a trading agent as well as for a general worst case methodology where {\it each relevant trajectory counts}.

A natural data structure for representing our trajectory sets, to be described in this section, is a directed graph. During the process of constructing possible
children at a given node $({\bf X},k)$, 
there will be two main steps: 1- children nodes are proposed, 2- some of proposed children may be pruned, i.e. deleted since they do not satisfy certain conditions. 
The first step, the proposal stage, will make $({\bf X},k)$ a \textit{no-arbitrage} (i.e. arbitrage-free) node but the pruning
process may break this property as the two stages work independently. This phenomena we have found to be possible in our algorithmic construction but less likely 
if the pruning is risk averse. More specifically, pruning will be controlled by worst case parameters and the latter can be modified if there is a conscious intention to introduce risk to increase potential profits.

We have adopted a trajectory set construction method that is agent-based and operational (see \cite{ferrando19} for a justification). \\
Nonetheless, it should be clear that our approach is one possibility out of many, in particular, the general scenario-based and superhedging framework is specially suited to
deploy constructions by means of machine learning patterns.

Whenever relevant, model variables will be capitalized, e.g. $X$, while the corresponding observable quantity will be denoted with lower-case letters, e.g. $x$.
The following sections are organized to not only define the components of the models but also provide a rationale for their construction. We also illustrate via software output several aspects and properties of the said constructions.

\subsection{Operational Data Processing}\label{sec:operational_data_processing}
We introduce the framework used to employ an operational approach when constructing trajectory market models. We follow a discretized approach for each quantity and variable and hence, through a slight abuse of notation, we use the notation for an interval $[a,b]$ to represent $[a,b] \equiv [a,b] \cap \Delta \mathbb{Z}$ where 
\[\Delta\mathbb{Z}\equiv\{\hdots, -2\Delta, -\Delta, 0, \Delta, 2\Delta, \hdots\}.\] 
We observe time intervals in a discrete sense, with a smallest possible time resolution $\Delta>0$ (usually in units of minutes). The investor can only observe the market in increments of $\Delta$. We consider historical times to be negatively valued, $0$ being the present time and future times to be positive. Let $\Te$ denote the entire observable time interval of the past.  $\Te$ takes the form $\Te=\{\hdots, -3\Delta, -2\Delta, -\Delta,0\}$.\\

Within the entire historic time interval $\mathcal{T}$ we will consider non-overlapping subsets $[t_0, t_0+T]\subseteq\mathcal{T}$ where $T\equiv M_T\Delta>0$ and $M_T$ is the number of $\Delta$ increments occurring in $[t_0, t_0+T]$, i.e. $[t_0,t_0+T]\equiv\{t_0, t_0+\Delta, t_0+2\Delta,\hdots, t_0+M_T\Delta\}.$ 
In our specific trading setting, $[t_0, t_0+T]$ refers to a day's worth of trading time but any interval could be represented as $[t_0, t_0+T]$ in general. Given that $T$ is a fixed quantity, we sometimes write $I_{t_0}\equiv [t_0,t_0+T]$   as a shorthand, with $t_0\in\{\hdots,-3T-2, -2T-1,-T\}$. We often refer to $\It$ as a \emph{time interval}, \emph{time window} or simply \emph{window}. 
We also make use of the disjoint collection of time intervals, namely: 
\begin{equation} \nonumber
 I=\{\hdots,[-3T-2,-2T-2],[-2T-1,-T-1],[-T,0]\} = 
\end{equation}
\begin{equation} \nonumber 
 \{\hdots, I_{-3T-2}, I_{-2T-1}, I_{-T} \}=\{\It: t_0\in\{\hdots, -3T-2, -2T-1,-T\}\}.    
\end{equation}  


Let \textit{undiscounted charts} refer to  historical market quoted prices: 
$s(t) = \big(s^0(t), s^1(t),s^2(t)\big)$ at some time $t\in \Te.$ 
We may obtain discounted prices with respect to a numeraire  by means of taking a ratio, through which  we get \textit{discounted charts} or simply \textit{charts}:
\begin{equation}\nonumber
x^j(t)\equiv \frac{s^j(t)}{s^0(t)},\ \ \  j=0,1,2,\\
\end{equation}
and $x(t)=(x^1(t),x^2(t)).$ When referring to  entire charts we denote $x=\{x(t): t\in\mathcal{T}\}$ and $x^j=\{x^j(t): t\in\mathcal{T}\}$. There is also another chart $x^0(t)\equiv 1$ which remains constant for all $t\in\mathcal{T},$ however this chart is not an explicit part of our construction/notation.  

It is assumed that both discounted and undiscounted charts move in discrete integer multiples of a minimal unit change. Undiscounted charts move in increments of some smallest unit of currency per unit of asset, while the issue is more ambiguous for discounted charts given that the quotient of two quantities may be any rational number, i.e. there is no smallest unit of numeraire. In practice the investor observes the entire historical data and uses their own methodology to determine the discretization parameters. \\

\subsection{$\delta$-Escapes and $\delta$-Escape Times} \label{deltaEscapeTimes}

This section prescribes how our agent rebalances her portfolio in response
to the change of chart values. Such a setup we term operational to emphasize that
rebalances are associated to observable variables. This association represents the main driving force
behind our model construction and, as argued in \\
\cite{ferrando19}, it makes our models objective. That is, there is a way to relate the unfolding chart values to the associated modelling trajectory. 
There is, of course, a plethora of ways to prescribe such operational setup and 
our proposal is a possible one. For this reason, we try to emphasize some
general aspects of our construction so that the arguments can be adapted
to different settings.

Martingale processes underline much of financial price modelling due to the first fundamental theorem
of asset pricing (\cite{follmer}). It is therefore relevant to rely on some of their properties
when prescribing our operational setting. Remarkable martingale properties are given by some type of regularity of their oscillations, this is reflected in inequalities for the number of upcrossings, e.g. Doob's and Dubin's upcrossing inequalities (\cite{neveu}) or Burkholder's inequality 
for the number of escapes of a martingale (\cite{burkholder}). Recent results on non-probabilistic martingales, which are closely related to  no-arbitrage considerations, 
like the ones in \cite{shafer} (see also \cite{bender3}) indicate that the mentioned results are probability independent and so represent a phenomena of greater scope. These comments should
be seen as motivation for the definitions that follow.

$\delta$-escapes and $\delta$-escape times, are introduced below for two models, henceforth denoted by A and B. Definitions assume, implicitly, a given chart $x$ as well as a given time interval $I_{t_0}=[t_0, t_0+T]$.

\begin{definition}[Model A] \label{modelA}
For given parameter  values $\delta^{0}, \delta^1>0$ and $t_0 \leq t' < t \leq t_0+T$, define the following $\delta$-escapes as
\begin{equation}\label{eqn::deltaescapes1}
\delta^{A,0}_{escape}(x,\It,t,t')\equiv  {\lvert x^1(t)-x^1(t')\rvert}
\end{equation} 
\begin{equation} \nonumber 
\delta_{\text{escape}}^{A,1}(x,\It,t,t')\equiv  \frac{\lvert x^2(t)-x^2(t')\rvert}{\lvert x^2(t') \rvert},
\end{equation}
For $i\geq 1,$ define the $i$-th $\delta$-escape time recursively to be 
\begin{equation} \nonumber 
\begin{split}
t_{i}\equiv\min\{~t:~ t_{i-1} < t\leq t_0+T:\quad ~ \delta^0\leq \delta_{escape}^{A,0}(x,\It,t,t_{i-1}) \quad\text{or}\\ \quad\quad\quad\quad \delta^1\leq \delta_{escape}^{A,1}(x,\It,t,t_{i-1})\}
\end{split}  
\end{equation}
whenever the set on the right-hand side is non-empty, otherwise $t_i$ is left undefined.  
\end{definition}

\begin{definition}[Model B] \label{modelB}
For given parameter  value $\delta^{B} > 0$ and $t_0 \leq t' < t \leq t_0+T$, define the following $\delta$-escape as
\begin{equation} \nonumber 
\delta_{\text{escape}}^B(x,\It,t,t')\equiv\max \left( \frac{\lvert x^2(t)-x^2(t') \rvert}{\lvert x^2(t') \rvert},\frac{ \lvert x^1(t))-x^1(t')\rvert}{\lvert x^1(t')\rvert}\right).\\
\end{equation}
For $i\geq 1,$ define the $i$-th $\delta$-escape time recursively to be 
\begin{equation} \nonumber 
    t_i\equiv\min\{~t:~ t_{i-1} < t\leq t_0+T\quad\text{and}\quad ~ \delta^B\leq \delta_{escape}^B(x,\It,t,t_{i-1}) \}
\end{equation}
whenever the set of the right-hand side is nonempty, otherwise $t_i$ is left undefined. For $i=1$, set $t_{i-1}\equiv t_0$. 
\end{definition}

For either model, we say that the $i$-th  $\delta$-escape occurred if $t_i$ is defined and let $N$ denote the largest integer $i$ such that $t_i$ is defined. If $t_1$ is not defined set $N=0$. 
Then, $\{t_i\}_{0\leq i\leq N}$ is the \emph{sequence of $\delta$-escape times} and $N$ is the \emph{number of $\delta$-escape times} occurring within the time interval $[t_0, t_0+T].$  
Notice that $t_0$ is the first element of the time interval  and it is always included in the sequence, and $t_0<t_1<\hdots < t_N\leq t_0+T$. For future use we also set $t(I_{t_0}) \equiv \{t_i\}_{0 \leq i \leq N}$ for either of the two models. We will use the notation $N= N(x, I_{t_0})$.

\subsection{Discretization}
While the notation $x^j$ is originally used to represent exact asset values, we need to contend with discretized variables when building our trajectory sets. 
For $j=1,2,
$ we let $\hat{\delta}^j$ be the  \emph{discretization parameter} for discounted chart $x^j$ and let $\lfloor\cdot\rfloor$, $\lfloor\cdot\rfloor_{\hat{\delta}^j}$  denote rounding to the nearest integer and rounding to the nearest integer multiple of $\hat{\delta}^j$, respectively. We may then define $k^j(t)\in\mathbb{Z}_+$  by  $\lfloor{x}^j(t)\rfloor_{\hat{\delta}^j}=k^j(t)\hat{\delta}^j$. Of course, $\lfloor{x}^j(t)\rfloor_{\hat{\delta}^j}$ may or may not represent exact historical values anymore, depending on the choice of numeraire $s^0$ and the investor's choice of $\hat{\delta}^j$. The integers $k^j$ represent the discretized movement of the $j$-th asset. 
Occasionally we place the time variable as an index in order to emphasize its discrete nature, in that case we will introduce 
$k_{n \Delta}^j$  such that $\lfloor{x}^j(t)\rfloor_{\hat{\delta}^j}= k_{n \Delta}^j \hat{\delta}^j$ and $k^j(t) = k_{n \Delta}^j$ with $t =t_0+ n \Delta$, $0 \leq n \leq M_T$, where $[t_0, t_0+T]$ is given. The above discretization methodology becomes actually exact when data  has a smallest discrete unit, this is the case of prices listed in terms of currency.

\subsection{Variation}  \label{sec:Variation}
We define next \emph{accumulated variation} (or, as a shorthand, \emph{variation}) to be a convenient  representation of a chart accumulated historical movements. This is one of the additional coordinates $Z_i$ that were
introduced in an abstract way in Definition \ref{trajectories}.
At a  given time interval $\It=[t_0, t_0+T]\in I$, for $t\equiv t_0 + n~\Delta\in[t_0, t_0+T]$ with $n \in \mathbb{Z}_+, ~0 \leq  n \leq M_T$, we have the set of discretized asset  values $\lfloor{x}^j(t)\rfloor_{\hat{\delta}^j}= k^j_{n \Delta}~\hat{\delta}^j, ~j=1,2$ (or the alternative notation $k^j(t) \equiv k^j_{n \Delta}$). By means of these quantities,  we define the accumulated variation:
\begin{equation}\label{w}
    w(t) =\sum_{n=0}^{v-1}\lvert k_{(n+1)\Delta}^1-k_{n\Delta}^1\rvert+\lvert k_{(n+1)\Delta}^2-k_{n\Delta}^2\rvert,~\mbox{where}~~~t= t_0+ v \Delta, ~~
\end{equation}
and $w(t_0)=0$. Thus, $w(t)$ aggregates, as $t$ evolves on $[t_0, t_0+T]$, the combined integer movements of both charts without explicit reference to any discretization parameter, i.e. $w(t)\geq 0$ is integer valued and unitless. Given that we are primarily interested in modelling future asset values, the variation is introduced primarily as an instrument for pruning, i.e. by introducing the  empirically measurable variable $w(t)$, we are able to further restrict future scenarios via constraints available from a worst case perspective of the past (see Section \ref{pruning1} and Appendix \ref{additional_pruning_constraints} for details on constrains). From this point of view, it is clear that other quantities
could be used as well (e.g., a discrete version of quadratic variation).
We later  introduce (see  Appendix \ref{additional_pruning_constraints}) various definitions of accumulated variation as a function of other variables as needed, however, they are all derived from (\ref{w}).

\subsection{The Empirical Set $N_E$} \label{sec:theEmpiricalSet}
A key definition in our model construction is the  empirical set $N_E$ introduced below. It is the set of relevant (to the investor) joint observable price increments that occur at all $\delta$-escape times and gathered over all historical time windows. We build the set of empirical changes $N_E(x,I_{t_0})$ for the time window $\It$ in the following way: moving along a fixed time interval $I_{t_0}=[t_0, t_0+T]$, we recursively build the $\delta$-escape times $\{t_i\}_{0\leq i\leq N}$, according  to  Definition \ref{modelA} (for Model A) or Definition \ref{modelB} (for Model B)
and collect vectors $\left(\Delta_{t_i}x^1, \Delta_{t_i}x^2, 1,\Delta_{t_i}t, \Delta_{t_i}w\right)$ 
representing variable increments between  $\delta$-escape times. 
More precisely,  the following various increments are defined for $0\leq i\leq N-1$ and for each of the two models:
\begin{equation*}
    \lfloor\Delta_{t_i}x^j\rfloor_{\hat{\delta}^j} \equiv \lfloor x^j(t_{i+1})\rfloor_{\hat{\delta}^j}-\lfloor x^j(t_i)\rfloor_{\hat{\delta}^j}=(k_{t_{i+1}}^j-k_{t_i}^j)\hat{\delta}^j
    \equiv m_i^j\hat{\delta}^j, \quad j=1,2
\end{equation*}
\begin{equation*}
    \Delta_{t_i}t\equiv t_{i+1}-t_i
    \equiv q_i \Delta
\end{equation*}
\begin{align*}
    \Delta_{t_i}w&\equiv w(t_{i+1})-w(t_i)=\sum_{n=u}^{v-1}\lvert k^1_{(n+1)\Delta}-k^1_{n\Delta}\rvert +\lvert k^2_{(n+1)\Delta}-k^2_{n\Delta} \rvert \equiv\eta_i,
\end{align*}
where in the last line, $t_i\equiv u\Delta$ and $t_{i+1}\equiv v\Delta$ for some $0\leq u< v\leq M_T, ~u,v\in\mathbb{Z}_+.$ If $N=0$ then all of the aforementioned variable increments are set to $0$. 
Notice $m_i^j, q_i, \eta_i$ are all integers (non-negative in the case of $q_i$ and $\eta_i$), and variation has no discretization parameter (and is also, as a consequence, unitless).

Hence the set of chart changes over the window $\It$ is the collection of all such vectors, rounded to their nearest discretization parameters:
\begin{equation} \nonumber
N_E(x,\It)=\{\left(m_i^1, m_i^2, 1, q_i, \eta_i\right) \equiv ,\:\
\end{equation}
\begin{equation} \nonumber
\left(\frac{\lfloor\Delta_{t_i} x^1\rfloor_{\hat{\delta}^1}}{\hat{\delta}^1}, \frac{\lfloor\Delta_{t_i} x^2\rfloor_{\hat{\delta}^2}}{\hat{\delta}^2}, 1, \frac{t_{i+1}-t_i}{\Delta}, \Delta_{t_i}w\right)~t_i \in t(\It)\},
\end{equation}
where the notation $t(\It) = \{t_i\}_{\{0 \leq i \leq N\}}$ was introduced at the end of Section \ref{deltaEscapeTimes} and  $N=N(x,\It)$. 
The third coordinate comes from the fact that $1=(i+1)-i.$ Notice that the size of $N_E(x,\It)$ is determined by the number of $\delta$-escapes, i.e. whenever $N\geq 1$ i.e. $\lvert N_E(x,\It)\rvert =N\equiv N(x,\It).$ In the case where $N=0$,  $N_E(x,\It)$ consists of the $(0,0,1,0,0)$ vector.


The \textit{set of empirically measured chart changes} is the collection of all such vectors $\left(\frac{\lfloor\Delta_{t_i} x^1\rfloor_{\hat{\delta}^1}}{\hat{\delta}^1}, \frac{\lfloor\Delta_{t_i} x^2\rfloor_{\hat{\delta}^2}}{\hat{\delta}^1}, 1,\frac{t_{i+1}-t_i}{\Delta}, \Delta_{t_i}w\right)$, $t_i \in t(\It)$  over all time windows $\It\in I$. For clarity, we define this set to be:

\begin{equation}\nonumber
   N_E\equiv N_E(x,I)\equiv\bigg\{\left(m_i^1, m_i^2, 1, q_i, \eta_i\right)\in N_E(x, I_{t_0}): \It\in I,  \bigg\},  
\end{equation}
(where $I$ was introduced in Section \ref{sec:operational_data_processing}) with the size of $N_E$ being given by $\lvert N_E \rvert= \lvert \cup_{\It\in I} N(x,\It) \rvert.$ Figure \ref{fig:good_N_E} demonstrates a convex hull under model B.
We indicate that, as a result of empirical observation, the vertices of the convex 
hull generated by  the first two coordinates of $N_E$ are the ones that seem to affect the value of $\overline{\sigma}_i$ at least in the current setting. It should be clear that this fact, whenever holds, could be used to compute
prices very efficiently by the backwards dynamic algorithm described in Appendix \ref{sec:Pricing}  .  

\begin{figure}
    \centering
    \includegraphics[scale=.5]{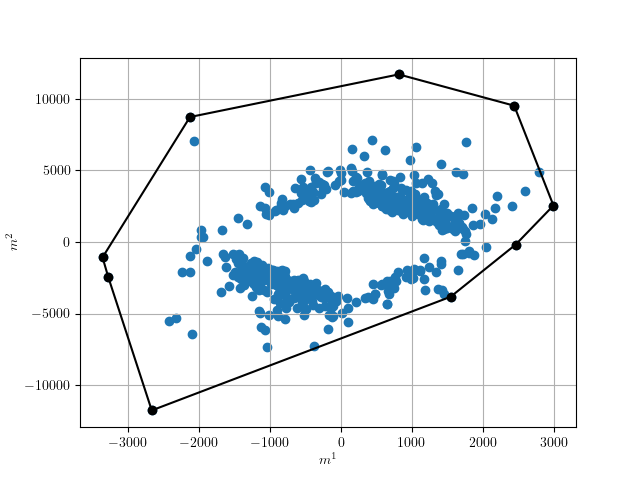}
    \caption{The first two components $(m^1, m^2)$ of set $N_E$ constructed under model B from historical data, with $\delta^B=0.011.$ 
     Black lines denote the convex hull of this set, with black points denoting the vertices.}
    \label{fig:good_N_E}
\end{figure}

\subsection{Pruning Constraints} \label{pruning1}
In order to limit the growth of the trajectory set and, more substantially, to have our trajectories reflect more closely the past historical data, we introduce \emph{pruning constraints} or \emph{pruning functions}. These constraints are determined by the variables which are used as additional modelling coordinates and are an essential feature of the models. The most suitable selection of variables will give tight (in terms of lower and upper bounds) and stable (in terms of historical data aggregation) worst case constraints. Moreover, the variables for our models should be chosen so as to complement each other; each variable ought to pick up different characteristics. 
The purpose of the pruning functions is to constrain different variables to a smaller set of possibilities; we proceed with a worst case methodology, namely, {\it we construct trajectory models that do not contain worst case scenarios not appearing in historical data}.
Section \ref{smallArbitrages} provides an alternative, theory based, approach to pruning.

Each constraint is a pair of functions, a maximum and minimum quantity denoted by $^*$ and $_*$ respectively, constructed by relying on a given chart $x$ and all historical windows $\It\in I$. A third argument for each constraint represents one of the following modelling variables: time $\rho$, number of $\delta$-escapes $i$ or variation $w$.  The pruning constraints will depend on the model being chosen, being either A or B; this dependency is implicit in the definitions (e.g. through the dependency on rebalancing times).

\vspace{.1in}
We provide one example of a pruning constraint in Definition \ref{def::typeIpruningconstraint_N_vs_T} below, other examples are relegated to Appendix  \ref{additional_pruning_constraints}.

\begin{definition}[Historical Maximum and Minimum Number of $\delta$-Movements at Time $\rho$] \label{def::typeIpruningconstraint_N_vs_T}
For a given chart $x,$ time interval $\It$, portfolio rebalances times $t(\It)=\{t_i\}_{0\leq i\leq N}$ and $\rho\in\{0,\Delta, \hdots, M_T\Delta\},$ define the number of $\delta$-movements in the time interval $[t_0, t_0+\rho]\subseteq I_{t_0}$ by:
\begin{equation} \nonumber 
 N(x,\It,\rho)= \max_{0\leq i\leq N}\{~i :~ t_i \leq t_0+\rho\}.
\end{equation}
Then, the corresponding Maximum and Minimum number of $\delta$-movements at time $\rho$ are defined as: 
\begin{equation}\label{eqn::typeIpruningconstraint_N_vs_T}
N^*(x,I,\rho) = \max_{I_{t_o}\in I}\:\: N(x,\It,\rho), \quad N_*(x,I,\rho) = \min_{I_{t_o}\in I}\:\: N(x,\It,\rho)
\end{equation}
for $\rho\in\{0,\Delta, \hdots, M_T\Delta\}.$
\end{definition}
While constructing model trajectories,  the
proposed future trajectory segment, with time coordinate $T_i = \rho$, will be discarded if the following constraint is not satisfied.
$N_*(x,I,\rho) \leq i \leq N^*(x,I,\rho)$.
For illustration, the effects of pruning are demonstrated in Figure \ref{fig:pruning}
while Figure \ref{fig:numberOfRebalances} displays the minimum and 
maximum number of rebalances.
\begin{figure} 
    \centering
        \includegraphics[scale=.65]{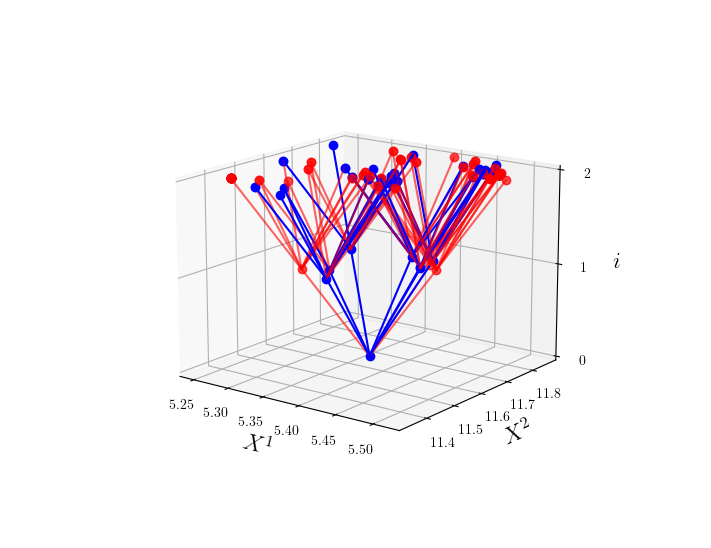}
    \caption{Trajectory set consisting of nodes $(X_i^1, X_i^2, i)$ up to $N({\bf X})=2$ where red trajectories (consisting of nodes and edges) are pruned according to pruning constraints and blue trajectories which remain after pruning (such remaining nodes will be called {\it admissible}).} 
\label{fig:pruning}     
\end{figure}

\begin{figure}
    \centering
    \includegraphics[scale=.60]{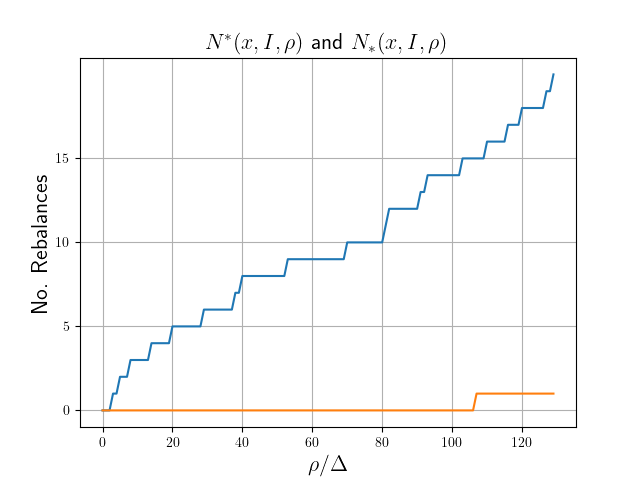}
    \caption{A pair of pruning constraints. $\delta^B=0.15$ (as per Definition \ref{modelB}), $\hat{\delta}^1=\hat{\delta}^2=0.01$ (as introduced in Section \ref{sec:theEmpiricalSet}). $\rho/\Delta\in\{0, \hdots, M_T\}$ and $M_T=130.$ In this example, the pruning constraints are monotone non-decreasing, although in general this may not be  the case.}
    \label{fig:numberOfRebalances}
\end{figure}

\section{Model Specification} \label{modelSpecification}
This section makes explicit the trajectory sets $\mathcal{X}$ generated under models A and B. We call elements of $\mathcal{X}$ \emph{model trajectories}, they are constructed recursively by relying on historical values of observed charts. Models A and B are differentiated by the way $\delta$-escape times are defined and the repercussions of this fact (e.g. the sets $N_E$, pruning constraints, etc) otherwise, the method of trajectory generation is identical in both cases. 

A trajectory set $\mathcal{X}$ consists of trajectories ${\bf X}$ i.e. sequence of multidimensional vectors ${\bf X} \equiv \{{\bf X}_i\}_{i\geq 0}\in\mathcal{X}$  where ${\bf X}_i = (X_i, Z_i) \equiv (X^1_i, X^2_i, i, T_i, W_i)$ is often referred to as a \emph{node}. A trajectory may also  be thought of as a set of nodes connected by (directed) edges. The empirical historic chart counterparts are denoted in lowercase i.e. the historic charts $x^1, x^2$ and $t_i, w_i$. $X^1_i, X_i^2$ represent model asset values, $T_i$ represents model time values at the $i$'th  $\delta$-escape time, and $W_i$ represents the (modelled) accumulated variation of the two dimensional vector $X_i\equiv (X_i^1, X^2_i)$ at the $i$'th  $\delta$-escape time.
These associations between model and historic data  indicate that our trajectories lie within a discrete grid of points based on the historic time parameter $\Delta$ and the investor calibrated parameters $\deltahatone$, $\deltahattwo$: ${\bf X}_i \equiv (X^1_i, X^2_i, i, T_i, W_i) \in (\deltahatone \mathbb{Z} \times \deltahattwo \mathbb{Z}\times \mathbb{Z}_+ \times \Delta \mathbb{Z}_+ \times \mathbb{Z}_+)$.

The empirical parameter $N$, introduced after Definitions \ref{modelA} and \ref{modelB},
will be represented in the model by an integer $N({\bf X})$. Our trajectories
will terminate after a finite number  of steps, i.e. ${\bf X} \equiv \{{\bf X}_i\}_{0 \leq i \leq N({\bf X})}$ the last (potential) trade takes place at instance 
$N({\bf X})-1$ along the model trajectory ${\bf X}$, the model coordinate
$T_{N({\bf X})}$ does not necessarily correspond to terminal time $T$.

We begin with the initial state ${\bf X}_0 = (X^1_0, X^2_0, 0, T_0, W_0)$ where $T_0=0, W_0=0$ and $X_0^1=x^1(0), X_0^2=x^2(0)$ are the most recent discounted chart values. Relying on an empirical set $N_E$, corresponding to a specific model A or B,  we generate trajectories recursively on the index $i$, referred informally as \emph{time steps}. More precisely, the index $i$ refers to the (potential) $i$-th.   \emph{portfolio rebalance} that takes place along the model's trajectory. That is, we are building values of trajectories not at  all possible times but at specific times associated to historical $\delta$-escapes, the latter triggering portfolio rebalances. For example, along a modelling trajectory, $T_i$ and $W_i$ tell us the modelling time and variation at the $i$-th. (potential) portfolio rebalance 
while $i$ tells us the number of portfolio rebalances so far. 

Given a node ${\bf X}_i =(X^1_i, X^2_i, i, T_i, W_i)$, $i\geq 0$,  we build trajectories by iteration on $i$  in the following way: for each $(m^1, m^2, 1,  q,\eta)\in N_E$ (the latter set was introduced in Section \ref{sec:theEmpiricalSet}) set ${\bf X}_{i+1}=(X_{i+1}^1, X_{i+1}^2, i+1, T_{i+1}, W_{i+1})$ where
\begin{align*}
X_{i+1}^1&=X_i^1+\Delta_i X^1= X_i^1+m^1\hat{\delta}^1,\\
X_{i+1}^2&=X_i^2+\Delta_i X^2= X_i^2+m^2\hat{\delta}^2,\\
T_{i+1}&=T_i+\Delta_iT= T_i+ q\Delta,\\
W_{i+1}&=W_i+\Delta_iW = W_i+ \eta.
\end{align*} \
This process is then continued in a recursive manner.
That is, we are recursively adding the set $N_E$ to the most recent node of each trajectory  and so the number of trajectories grows exponentially fast with at most $\lvert N_E \rvert^i$ new trajectories available at time step  $i.$ Figure \ref{fig:nice_CH_plot} illustrates the construction but only for the first two variables. 
 \begin{figure}
     \centering
     \includegraphics[scale=.6]{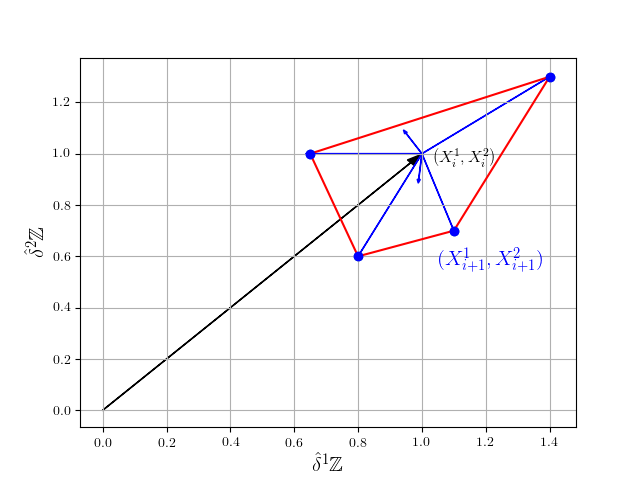}
     \caption{A two dimensional representation of how the children $(X^1_{i+1},X_{i+1}^2)$ (denoted in blue) are generated from $(X_i^1, X_i^2)$ (denoted in black). Notice $(X_i^1, X_i^2)$ is in the closure of the convex hull (denoted in red) of the generated children nodes $(X^1_{i+1},X_{i+1}^2)$. In this example $\deltahatone=0.01=\deltahattwo.$}
     \label{fig:nice_CH_plot}
 \end{figure}
 
\subsection{Dynamic Pruning}
Once all the possible future nodes ${\bf X}_{i+1}$ are built from  the current node ${\bf X}_i$ by means of the set $N_E$, we then check whether these future states are historically realistic, i.e. we prune ${\bf X}_{i+1}$ according to the pruning constraints introduced in Section \ref{pruning1} and  Appendix \ref{additional_pruning_constraints}. If ${\bf X}_{i+1}$ obeys the pruning constraints set by the investor, we say ${\bf X}_{i+1}$ is \emph{admissible}. Below, for future reference, we let $N_A({\bf X}_i)$ to  be the set of admissible nodes constructed recursively from ${\bf X}_i$ i.e:
 \begin{equation} \label{admissible}
     N_A({\bf X}_i)\equiv\Big\{{\bf X}_{i+1} \equiv(X_i^1+m^1\hat{\delta}^1,X_i^2+m^2\hat{\delta}^2, i+1, T_i+q\Delta,  W_i+\eta)~\text{is ~admissible}:
\end{equation} 
\begin{equation} \nonumber    
      \quad(m^1,m^2,1,q,\eta)\in N_E\Big\}
 \end{equation}

More specifically,  to check for admissibility we will rely on the notation: $X^*(i)$, $X_*(i)$, $N^*(T_i)$, $N_*(T_i)$, $N^*(W_i)$, $N_*(W_i)$, $T^*(i)$, $T_*(i)$, $T^*(W_i)$, $T_*(W_i)$, $W^*(i)$, $W_*(i)$, $W^*(T_i)$, and $W_*(T_i)$, to denote the evaluation of the maximum/minimum pruning constraints, introduced in Appendix \ref{additional_pruning_constraints}, at node ${\bf X}_i=(X_i^1, X_i^2, i, T_i, W_i)$ respectively (for simplicity, we are suppressing the arguments $x$ and $I$ used in Appendix \ref{additional_pruning_constraints}). Notice that, in contrast to the definitions in Appendix \ref{additional_pruning_constraints} (see also Section \ref{pruning1}), we are now evaluating the constraints not at all times $\rho$ but only at the generated times $T_i$. Similarly, whenever variation is an argument, we only evaluate at the specific  variations $W_i$.

Hence ${\bf X}_{i+1}\equiv
(X_i^1+\Delta_iX^1,X_i^2+\Delta_iX^2, i+1, T_i+\Delta_iT,  W_i+\Delta_iW)$
is admissible if ${\bf X}_{i+1}$ satisfies the following pruning constraints:
\begin{align*}
    \frac{\|(X^1_{i+1},X^2_{i+1}) - (X^1_0,X^2_0)\|}{\|(X^1_0,X^2_0)\|} &\in [X_*(i+1), X^*(i+1)]
    \nonumber\\
    (i+1) &\in [N_*( T_i+\Delta_iT]), N^*(T_i +\Delta_iT)]\\
    (i+1)& \in [N_*(W_i + \Delta_iW), N^*(W_i + \Delta_iW)]\\
    T_i+\Delta_iT &\in \big[T_*({i+1}), T^*({i+1})\big]\\
T_i+\Delta_iT &\in \big[T_*(W_i+\Delta_iW), T^*(W_i+\Delta_iW)\big]\\
W_i + \Delta_iW &\in [W_*(i+1), W^*(i+1)]\\
W_i + \Delta_iW &\in [W_*(T_i+\Delta_iT), W^*(T_i+\Delta_iT)]. 
\end{align*}
We will say that {\it dynamic pruning} is in effect when the above constrains are being enforced during the trajectory set construction.
\section{Data and Calibration}  \label{dataAndCalibration}

 Historical data is observed in increments of 3 minutes, with trading occurring daily between 9:30 am  to 4:00 pm EST, hence we set $\Delta=3~\mbox{minutes}$ and $M_T=130$ time increments occurring during the day,  which results in $T\equiv M_T\Delta=390$ minutes or six and a half hours of daily trading time.

We collected historic stock data of Twitter, Facebook and Netflix. The data was collected in increments of 3 minutes between 9:30am and 4:00pm between 2018-05-09 and 2018-10-15. 
\begin{figure}[!h]
    \centering
    \includegraphics[scale=.65]{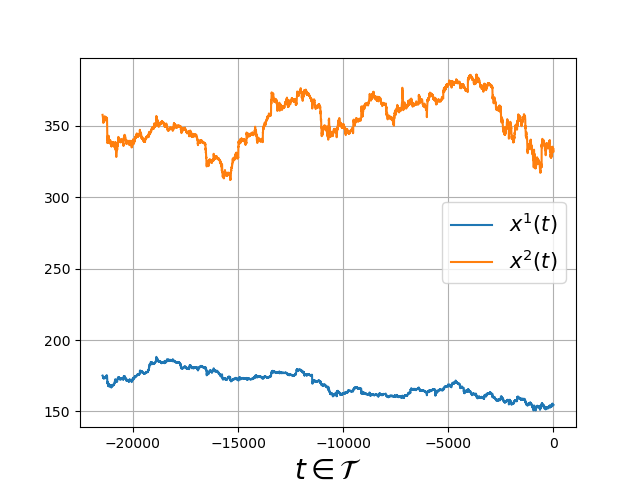}
    \caption{Historic data for charts  $x^2$ and $x^1$, where $x^2$ is the price of
the Netflix stock, $x^1$ is the price of the Facebook stock, both in units of the U.S. dollar. Times are given in increments of $\Delta$, and are negative to reflect that they are historical.}
\label{fig:stocks}
\end{figure}

\subsection{Calibration}
Before an investor can build pruning constraints, the empirical set $N_E$ and other historical estimations, they are required to select appropriate values for the parameters $\delta, \hat{\delta}^1$ and $\hat{\delta^2}$, a process which we refer to as {\it calibration}. These calibrated
parameters will have a direct effect on various outcomes of
our trajectory models. We have eliminated the need for many parameters found in \cite{crisci} for example, the comparable models in \cite{crisci} contain the following parameters which are required to be calibrated by an investor: $\delta, \delta_0, \deltahatone, \deltahattwo, \hat{\nu}_0$ (for definitions of $\delta_0$ and $\hat{\nu}_0$ see Section 3.1 and Section 5.2.1 in  \cite{crisci}). Our careful choice in selecting the methods for sampling $\delta$-escape times as well as redefining accumulated variation to be unitless leads us to require only the following list of parameters to be calibrated: $\delta, \deltahatone$ and $\deltahattwo$ (for Model A, $\delta$ refers to actually two parameters: $\delta^0$ and $\delta^1$, see Definition \ref{modelA}).

Given that models A and B are differentiated by the way they generate the $\delta$-escape times, we calibrate $\deltahatone, ~\deltahattwo$ to be the same for both models. In the case that we take the numeraire to the U.S Dollar, we set $\deltahatone, ~\deltahattwo=\$0.01$ 
 i.e. the smallest observed historical increment. Otherwise, when an investor is dealing with an asset numeraire, the choice of $\deltahatone$ and $\deltahattwo$ is not clear and will require an application-dependent calibration (this is a generic problem in price modelling, i.e. it is not specific to our methodology).

 In calibrating $\delta$ for each model A and B (in fact model A requires two such $\delta$-parameters) we demonstrate how the number of $\delta$-escapes impact other aspects of the models for a range of values for $\delta$.

A complete model description also requires that we calibrate $N({\bf X})$ i.e. the maximum number of $\delta$-escapes allowed for each of our modelled trajectories.
We note that, to reduce model risk, one wants to set $N({\bf X})$ to a historical minimum. For example, for model B, Figure \ref{fig:Calibration_Nstar} shows the maximum and minimum number of observed $\delta$-escapes over all historic trajectories are relatively constant, as a function of $\delta$, around the value of $\delta=0.011$. Within this region, the minimum number of observed $\delta$-escapes is found to be 3, meaning that for $\delta=0.011,$ \emph{all} historical trajectories have an observed number  of $\delta$-escapes which is not smaller than $3$. Therefore, to make sure that model trajectories match historical charts, we will require $N({\bf X}) \leq 3$. Notice that $N({\bf X}) < 3$ merely indicates that, in our model, trading stops before all $\delta$-escapes take place in the unfolding chart (which, implicitly, we are assuming to have at least $3$ $\delta$-escapes). That is, our model trajectories will have at most the historical minimum number of $\delta$-escapes. Small values of $N({\bf X})$ minimize modelling risk (as longer modelling trajectories are more likely not to reflect real trajectories). The choice  $\delta=0.011,$
gives some stability, which suggests a reduction on modelling risk, we also notice that the gap $N^*(x, I,\rho)- N_*(x, I,\rho)$ is being minimized with this choice and so pruning is maximized.

\vspace{.1in}
In general, stable and tight worst case bounds of a variable is an indication that the selected variable restricts the future manifold of possible trajectories. The actual values of the bounds and their stability also provide the possibility to calibrate to
realistic market conditions and to adjust modelling risk according to investing profiles. Illustrative output appears in Figures \ref{fig:Calibration_Nstar}, \ref{fig:my_labelLower} and 
\ref{fig:my_labelUpper}.

\begin{figure}[!h]
    \centering  
    \includegraphics[scale=.65]{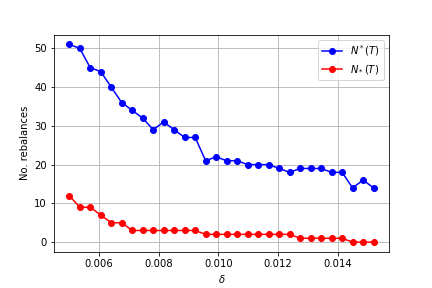}
    \caption{$N^*(x,I,\rho)$ and $N_*(x,I,\rho)$ over a range of values for  $\delta$ at $\rho=M_T\Delta=T$. Notice the relative stability for both $N^*(T)$ and $N_*(T)$ in the region $0.01\leq\delta\leq 0.012.$}
    \label{fig:Calibration_Nstar}
\end{figure}

\begin{figure}[!h]
    \centering
    \includegraphics[scale=.6]{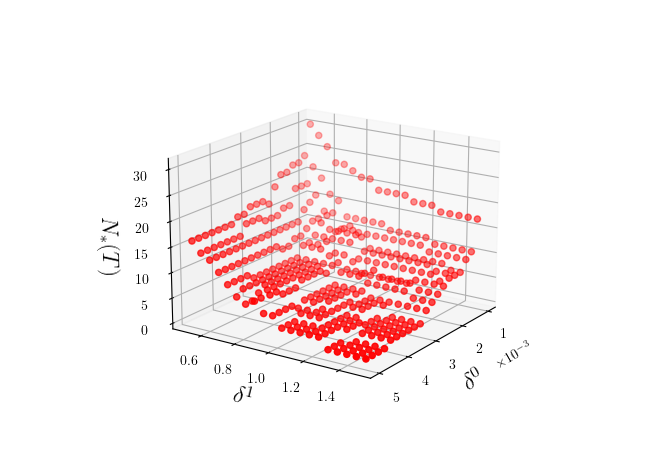}
    \caption{$N_*(x,I,\rho)$ at $\rho=M_T\Delta=T$ over a range of values for  $\delta^0$ and $\delta^1$ (Model A). Notice the relative stability of  $N_*(T)$ in the rectangle $0.002\leq \delta^0\leq 0.004$ and $0.8\leq\delta^1\leq 1.2.$}
    \label{fig:my_labelLower}
\end{figure}

\begin{figure}[!h]
    \centering
    \includegraphics[scale=.6]{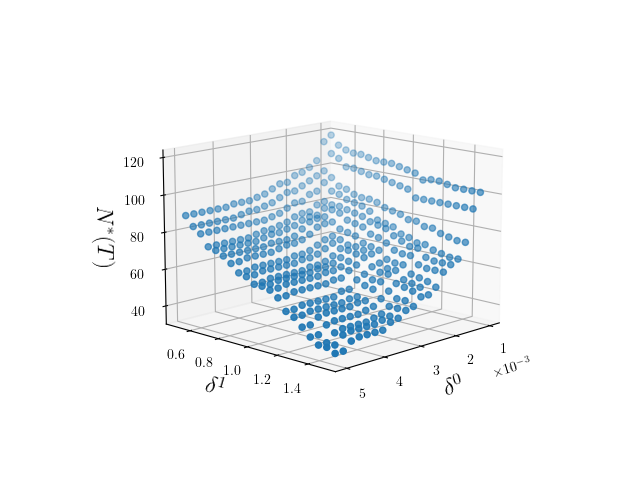}
    \caption{$N^*(x,I,\rho)$ at $\rho=M_T\Delta=T$ over a range of values for  $\delta^0$ and $\delta^1$ (Model A). Notice the relative stability of  $N^*(T)$ in the rectangle $0.002\leq \delta^0\leq 0.004$ and $0.8\leq\delta^1\leq 1.2.$}
    \label{fig:my_labelUpper}
\end{figure}

\clearpage
\pagebreak[4]
\subsection{Using Geometric Brownian Motion as Data}
The purpose of this section is to simulate longer periods of data than the empirical periods that we had access to. We do this in order to study the stability of various objects in our modelling approach, such as the convex hull of $N_E$ and pruning constraints.

In order to bypass the limited availability of historic data as well as an interesting conceptual exercise, we simulate stock prices over up to 5 years and observe the long-term behaviour. Figure \ref{fig:BM_convexhull_growth} demonstrates the shape of the convex hull of the two-dimensional set of points $(m^1, m^2)$ derived from $N_E$ for Model B. Figure \ref{fig:BM_pruning_growth} demonstrates the effects of increasing the aggregation of data on one pair of pruning constraints. In particular, neither the shape of the convex hull nor the bounds of the pruning constraints are affected significantly for large amounts of data.
\begin{figure}[!h]
    \centering
    \includegraphics[scale=.65]{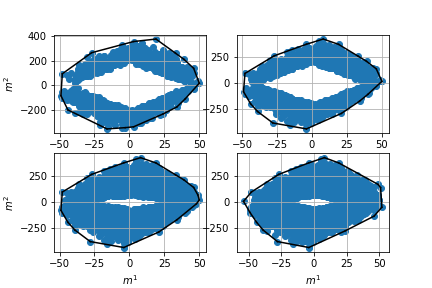}
    \caption{Growth of the convex hull, over (left to right/up to down) 6 months, 1 year, 2 years and 5 years. $\delta=0.014$, $\hat{\delta}^1=\hat{\delta}^2=0.01.$}
    \label{fig:BM_convexhull_growth}
\end{figure}
\begin{figure}[!h]
    \centering
    \includegraphics[scale=.65]{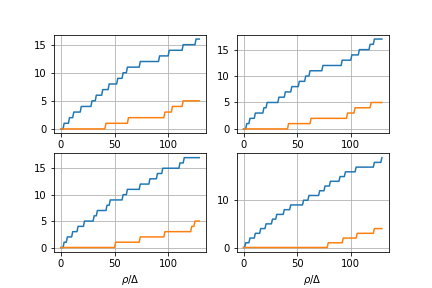}
    \caption{Pruning constraints $N^*(x,I,\rho), N_*(x,I,\rho)$, over (left to right/up to down) 6 months, 1 year, 2 years and 5 years. $\delta=0.001$, $\hat{\delta}^1=\hat{\delta}^2=0.01.$}
    \label{fig:BM_pruning_growth}
\end{figure}

\subsection{Simulating Trajectories} 
Figures \ref{fig:graphTypePlot}, \ref{fig:X1_vs_Ti} and \ref{fig:X2_vs_Ti} provide some graphical illustrations of trajectory sets.
\begin{figure} [!h]
    \centering
    \includegraphics[scale=.5]{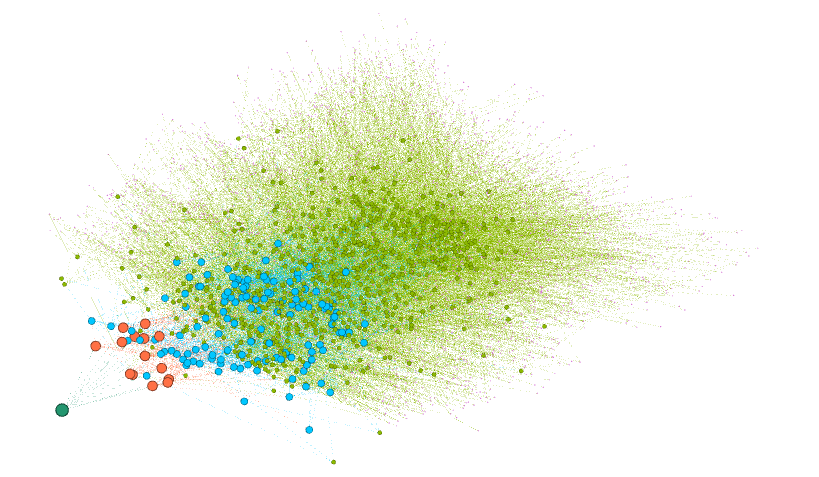}
    \caption{Trajectory set represented by a graph for the  B model with pruning. 
    Graph constains 4058 nodes and 11349 edges. $|N_E|=15$ and $N(X)=4.$}
    \label{fig:graphTypePlot}
\end{figure}

\begin{figure}[!h]
    \centering
    \includegraphics[scale=.60]{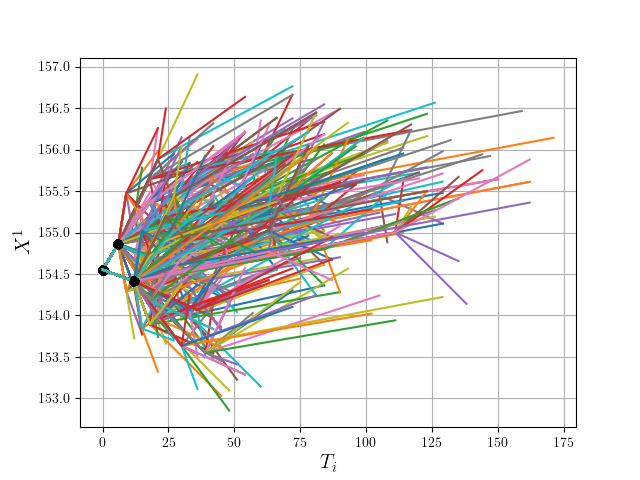}
    \caption{Trajectory simulations of $({\bf X}_i)_{0\leq i \leq N(X)}$ where ${\bf X}_i=(X_i^1, X_i^2, i,T_i,W_i)$. Figure depicts the coordinate $X_i^1$ as a function of $T_i.$}
    \label{fig:X1_vs_Ti}
\end{figure}
\begin{figure}[!h]
    \centering
    \includegraphics[scale=.60]{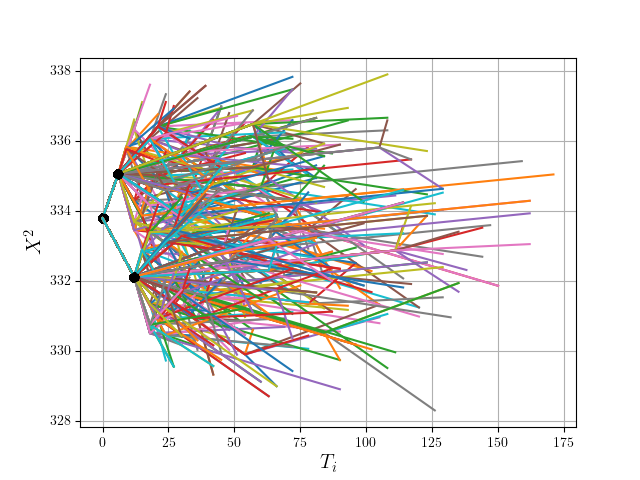}
    \caption{Trajectory simulations of $({\bf X}_i)_{0\leq i \leq N(X)}$ where ${\bf X_i}=(X_i^1, X_i^2, i,T_i,W_i)$. Figure  depicts the coordinate $X_i^2$ as a function of $T_i.$}
    \label{fig:X2_vs_Ti}
\end{figure}


\clearpage
\pagebreak[4]
\subsection{Trajectory Matching}  \label{trajectoryMatching}
This section provides some output illustrating how our model trajectories, that were built relying on historical charts, match future charts, i.e. price data not used for constructing the model's trajectories and used for testing purposes. We perform matching by recursively selecting a trajectory in our trajectory set that is a best match to the given chart. This matching is performed by minimizing a cumulative error over all coordinates. 
Given a chart $x=x(t)$, where $t\in\{0,\Delta, \hdots, M_T\Delta\}$, we calculate the $\delta$-escape times $\{t_i\}_{0\leq i\leq N}$ and collect the vectors $x_i\equiv (x_i^1, x_i^2, i, t_i, w_i)$ for $0\leq i\leq N$. We define next a matching trajectory $X^{match}\in\mathcal{X}$:
\begin{equation} \nonumber
    X^{match}=\arg~\min_{~X=(X_i)_{0\leq i\leq N({\bf X})}\in\mathcal{X}}~~ \{\sum_{i\geq 0}^{N({\bf X})} \| (X_i-x_i)\|_1\};
\end{equation}
we refer to it as the matching trajectory, with the corresponding error\\
 $\sum_{i\geq 0}^{N({\bf X})} \| (X_i^{match}-x_i)\|_1$.

\begin{figure}[!h]
    \centering
    \includegraphics[scale=.52]{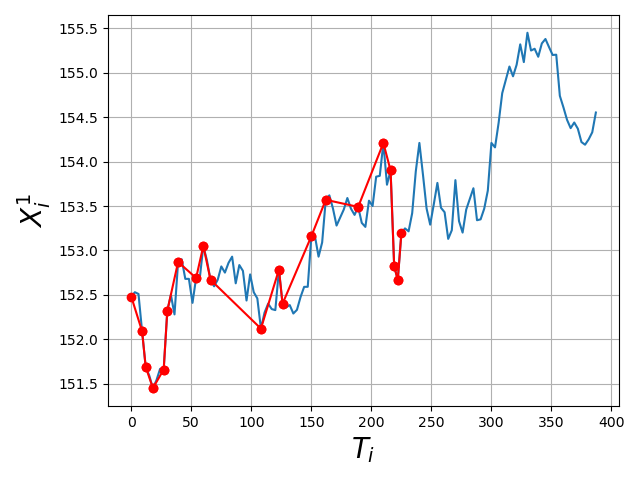}
    \caption{Matching the  chart (blue), with a trajectory (red) built using  the elements of $N_E$ for Model B.  In this case, for reference, the testing chart was used during the construction of the trajectory set. Total accumulated error amounted to $0.$ Trajectory matching occured up to $N({\bf X})=20$, with $\delta=0.011$.}
    \label{fig:X1_perfect}
\end{figure}
\begin{figure}[!h]
    \centering
    \includegraphics[scale=.50]{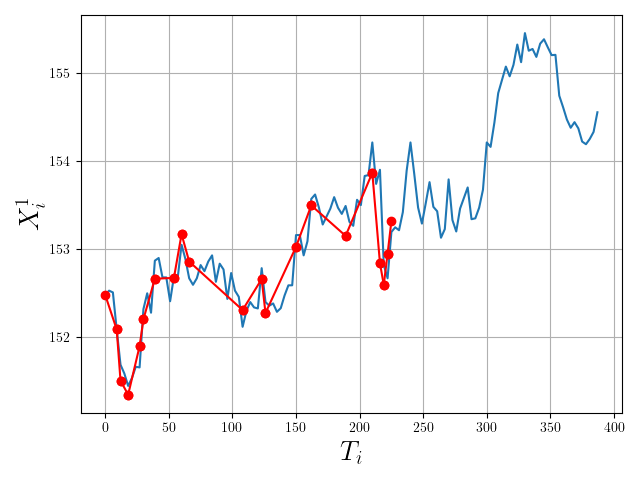}
    \caption{Matching a testing  trajectory (blue), with a trajectory (red) built using the elements of $N_E$ for Model B. The testing chart displayed in blue is the same as in Figure \ref{fig:X1_perfect} but this time it was not used during the construction of the trajectory set. Total accumulated error amounted to $11.4384.$ Trajectory matching occured up to $N({\bf X})=20$, with $\delta=0.011.$}
    \label{fig:match_X1_N_Eonly}
\end{figure}


\begin{figure}[!h]
    \centering
    \includegraphics[scale=.50]{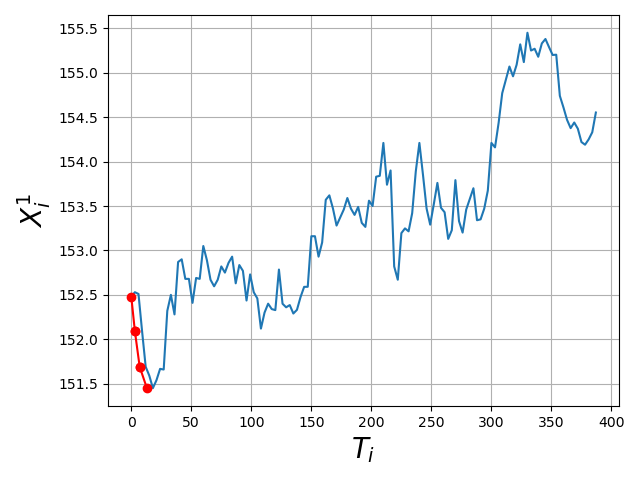}
    \caption{Matching the most recent historical trajectory (blue) built with a trajectory (red) built using the graph structure with dynamic pruning for Model B. The testing chart displayed in blue was not used during the construction of the trajectory set. Total accumulated error amounted to $1.2270.$ Trajectory matching occurred up to $N({\bf X})=4$, with $\delta=0.011.$}
    \label{fig:match_X1_G}
\end{figure}

\clearpage
\pagebreak[4]

\section{Profit and Loss Analysis} \label{sec:profitAndLossAnalysis}

This section elaborates on financial implications
derived from our superhedging models, we rely extensively on notation and definitions from Appendix \ref{theoreticalFramework}.  In particular, as detailed in the said appendix, we move freely from the $2$-dimensional trajectory model to the related $1$-dimensional trajectory set used for computations.

As indicated, each trajectory ends with a number of coordinates $N({\bf X})$, notice that 
$T_{N({\bf X})} < T$ is possible. That is, endings of  trajectories may take place
before the trading day ends. 

The models do produce pairs of possible
values $(X^1_{N({\bf X})}, X^2_{N({\bf X})})$ we then proceed to superhedge  the values
$X^2_{N({\bf X})}$ by trading with the asset $X^1$ and the numeraire. 

If we define $F({\bf X}) \equiv X^2_{N({\bf X})}$ (at times, for convenience, we may write this definition as $F= X^2$), from Corollary \ref{1DimensionalInBetweenBounds},  we have 
\begin{equation} \label{important}
\underline{\sigma}^1_{i} F({\bf X}) \leq X^2_i \leq \overline{\sigma}^1_{i} F({\bf X}), 
\end{equation}
for all $0\leq i \leq N({\bf X})$ and for some nodes $({\bf X}, i)$ (we refer to Remark \ref{validityOfL}, in Appendix \ref{theoreticalFramework}, for a discussion on the conditions needed to apply the said corollary as well as for the introduction of the notation $\overline{\sigma}^1_i F({\bf X})$).
Intuitively,  (\ref{important})  is a no-arbitrage result indicating that the model's prices $X^2_i$ can not be used to create a model arbitrage. In other words, the result shows that a trading strategy that involves short selling the asset $X^2$ and investing the proceeds into a portfolio always involves some risk
(i.e. the possibility to loose money along some trajectories). Let us provide some more precision, assume we are at a node $({\bf X}, i)$ where (\ref{important}) does not hold because $X^2_i > \overline{\sigma}^1_{i} F({\bf X})$. We then short
sell asset $X^2$ and invest in asset $X^1$ and the numeraire according to the definition in display (\ref{sigma_X_def}); it follows that:
\begin{equation}\nonumber
V_H(N(\hat{{\bf X}}), \hat{{\bf X}}) =  \overline{\sigma}^1_{i} F(\hat{{\bf X}})+ \sum_{k=0}^{N(\hat{{\bf X}})-1} H_k(\hat{{\bf X}})(\hat{X}^1_{k+1}- \hat{X}^1_k) \geq \hat{X}^2_{N(\hat{{\bf X}})}~~\mbox{for all}~~\hat{{\bf X}} \in \mathcal{X}_{({\bf X}, i)},
\end{equation}
more precisely, the above inequality holds up to a small $\epsilon$ and for an associated optimal portfolio $H$.
Our investor will then profit, for any conceivable model trajectory $~\hat{{\bf X}} \in\mathcal{X}_{({\bf X}, i)}$, at stage $N(\hat{{\bf X}})$.

\vspace{.1in}
In order to assess profit and loss properties  we proceed as follows: we evaluate a superhedging portfolio with the backwards pricing algorithm ,described in Appendix \ref{theoreticalFramework}, such portfolio, when fed with an initial investment $V= \overline{\sigma}_{0} X^2$, will satisfy $V_H(N({\bf X}), {\bf X}) \geq X^2_{N({\bf X})}$ for all ${\bf X}$ (again, up to a small $\epsilon$). Our numerical experiments will then
consider values of $V$ in the range $X^2_0 \leq V \leq  \overline{\sigma}^1_{0} X^2$, fed to the superhedging portfolio, and so introduce the possibility that the superhedging portfolio will not superhedge $X^2_{N({\bf X})}$ for all ${\bf X}$. This experiment then provides a profit and loss profile; for an initial investment in the said range, there will be some trajectories for which the superhedging property will fail. These sets of trajectories, where superhedging is uphold or where it fails, can be controlled in the model by modifying the level of pruning (and we provide output for different pruning approaches). In short,
risk in superhedging investment can be dosified in an objective way given that more or less pruning relates objectively to discarding or adding specific historical events.
 A dual experiment, where one purchases $X^2$ and short sells the underhedging portfolio, is also reported for underheging where one relies on the underhedging portfolio and a possible range of initial investments $V$ satisfying:  $\underline{\sigma}^1_{0} X^2 \leq V \leq X_0^2$. The underhedging portfolio is evaluated by the same backwards
pricing algorithm but, this time, the target is to superhedge $-X^2_{N({\bf X})}$. Therefore, evaluating the quantities $\overline{\sigma}^1_{i} (-X^2)({\bf X})$ one then obtains
$\underline{\sigma}^1_i X^2({\bf X})= -\overline{\sigma}^1_{i} (-X^2)({\bf X})$ and, similarly, the actual underheging portfolio is also obtained from the superhedging portfolio for $-X^2$
by multiplication by minus one.

 To quantify the level of risk that an investor must take- on when creating a portfolio of initial value $V$, we sample trajectories (uniformly) from the trajectory set $\mathcal{X}$, and determine the number of trajectories which profit in our model. Trajectory simulation is done by a recursive process, starting with the initial (common to all trajectories) node, and randomly selecting one of the (connected by an outgoing edge)  children nodes until we arrive at a node with no outgoing edges. One way of generalizing this exercise, which we do not explore, is to set a particular probability distribution onto $\mathcal{X}$, which affects the sampling of individual trajectories.
Not assuming a particular probability distribution, the notion of risk then refers to a set of trajectories (as opposed to a probability) where losses will take place.

  In other words,  we explore the consequences of an investor with an initial capital between the underhedging price and the superhedging price. We then simulate many trajectories along with various initial investments $V$ to investigate this type of risk-taking within our trajectorial market models.
 
The previous analysis neglects the risk of trajectory matching, i.e. it assumes the model trajectories will match exactly the market unfolding trajectory. A detailed analysis of this topic is presented in \\ \cite{ferrando19}) and some output illustrating trajectory matching is presented in Section \ref{trajectoryMatching}.

For simplicity, in our displays, we rely on the notation $X^2_{N({\bf X})}= F(X^1_{N({\bf X})})$ where
$F$ represents, necessarily, a multivalued function (i.e. a relation) implicit in the models' trajectory construction. We refer to $F(\cdot)$ as the payoff. This point of view is illustrated in Figure \ref{fig:my_label0}.

\begin{table}[!h]
\centering
 \begin{tabular}{||c ||c|c|l|} 
 \hline
   &  $~\underline{\sigma}X^2$ & $~\overline{\sigma}X^2$ & $X_0^2$\\
  
  \hline
  Model A & 321.9521 & 346.5000 &333.78\\ 
 \hline

 Model B & 324.5034& 341.3221 & 333.78\\
 \hline 
 \end{tabular}
\vspace{.1in} 
\caption{We calculate the price bounds $\overline{\sigma}X^2$ and $\underline{\sigma}X^2$ where $F(X^1_{N({\bf X})})=X_{N({\bf X})}^2$ and $N({\bf X})=3.$ } 
\end{table}
\begin{table}[!h]
\centering
 \begin{tabular}{||c ||c| c|c|c|l|} 
 \hline
 V  &$~\overline{\sigma}X^2$ &$X_0^2$& $X_0^2+1.00$ & $X_0^2-1.00$ &$\underline{\sigma}X^2$ 
 \\
 \hline
model A: &$100\%$& $37.5\%$& $49\%$   & $35\%$  &$0\%$\\
 \hline
 model B: & $100\%$& $37.5\%$& $50.5\%$   & $32\%$   & $0\%$
 \\[1ex] 
 \hline
\end{tabular}
\vspace{.1in}
\caption{For each model and initial investment, the percentage of trajectories which profit are recorded. In each case 1000 trajectories $(X_i)_{0\leq i \leq N({\bf X})}$  were simulated with $N({\bf X})=3$.  Notice that, as the initial investment approaches $\overline{\sigma}X^2$ a larger percentage of trajectories are expected to profit, and vice versa for  $\underline{\sigma}X^2$. Note $F(X^1_{(N({\bf X})})=X_{N({\bf X})}^2$, with $X_0^2=333.78$.  } \label{table:Prices}
\end{table}

\begin{figure}[!h]
    \centering
    \includegraphics[scale=.60]{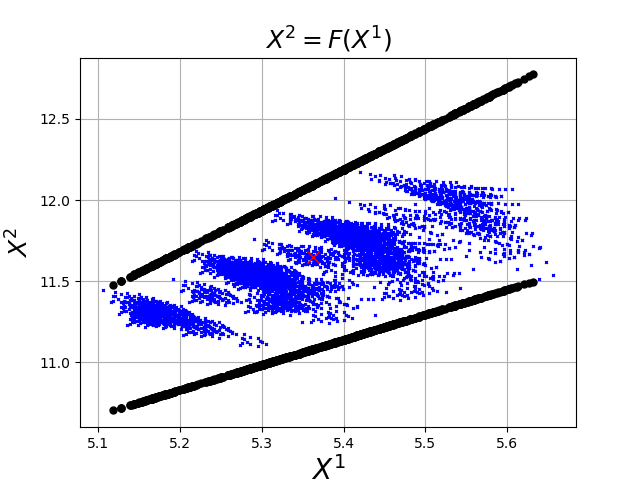}
     
    \caption{Simulated values of $(X^1_{N({\bf X})},X^2_{N({\bf X})})$ with black denoting portoflio values defined as $V+\sum_{i=0}^{N({\bf X})} H_i({\bf X})~(X_{i+1}^1-X_i^1)$ red marking denotes the initial node. Graph contains 5456 nodes and 14880 edges.}  \label{fig:my_label0}
   
\end{figure}

\begin{table}
\centering
 \begin{tabular}{||c ||c| c|c|c|l|} 
 \hline
 V  &$~\overline{\sigma}X^1$ &$X_0^1$& $X_0^1+1.00$ & $X_0^1-1.00$ &$\underline{\sigma}X^1$ 
 \\
 \hline
 model B: & $100\%$& $63.5\%$& $90\%$   &  $25\%$ & $0\%$
 \\[1ex] 
 \hline
\end{tabular}
\vspace{.1in}
\caption{For model B and initial investment, the percentage of trajectories which profit are recorded. Here, in contrast to the two previous tables, we superhedge $X^1$ by trading with $X^2$. In each case 1000 trajectories $(X_i)_{0\leq i \leq N({\bf X})}$  were simulated with $N({\bf X})=3$.  Notice as the initial investment approaches $\overline{\sigma}X^1$ a larger percentage of trajectories are expected to profit, and vice versa for  $\underline{\sigma}X^1$. Note  $F(X^2_{(N({\bf X})})=X_{N({\bf X})}^1$, with $X_0^1=154.5524.$ } \label{table:PricesofX1}
\end{table}

\begin{figure}[!h]
    \centering
    \includegraphics[scale=.5]{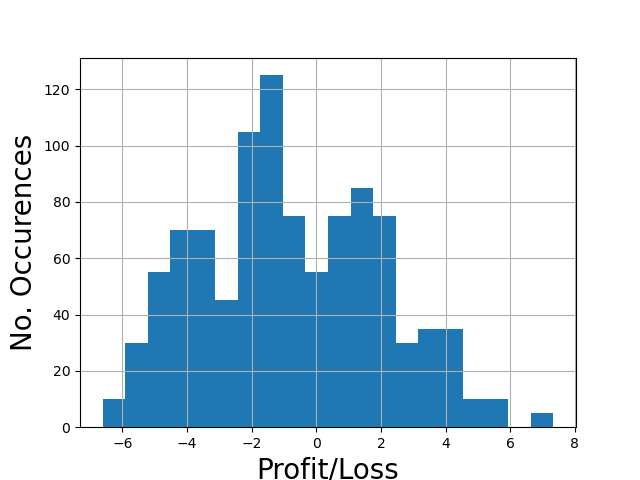}
    \caption{Profit/Loss histograms of 1000 randomly simulated trajectories, under model B  and payoff $X^2_{N({\bf X})}= F(X^1_{N({\bf X})})$ with initial investment $V=X_0^2=333.78$. The trajectory set is constructed with dynamic pruning (as introduced at the end of Section \ref{modelSpecification}) and parameters $\deltahatone=\deltahattwo=0.01$, $\delta=0.011$, and $N({\bf X})=3$.  
   ~~~ $37.5\%$ of trajectories profit.   
 Profit is given in units of USD.}
    \label{fig:V_is_X2}
\end{figure}
\begin{figure}[!h]
    \centering
    \includegraphics[scale=.5]{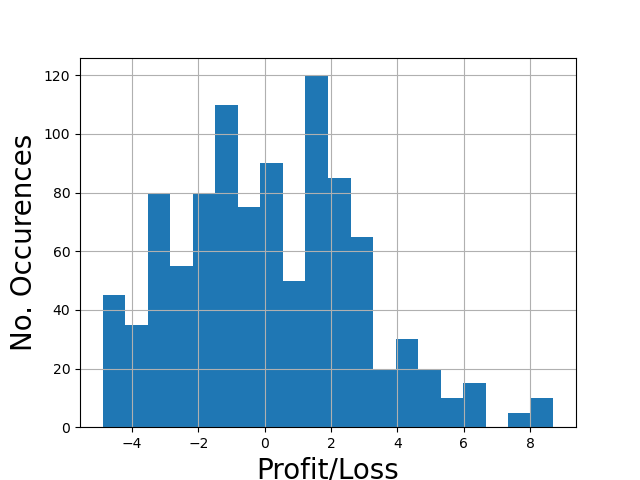}
    \caption{Profit/Loss histograms of 1000 randomly simulated trajectories, under model B and payoff $X^2_{N({\bf X})}= F(X^1_{N({\bf X})})$ with initial investment $V=X_0^2+v=333.78+1.00$. The trajectory set is constructed with dynamic pruning and parameters  $\deltahatone=\deltahattwo=0.01$, $\delta=0.011$, and $N({\bf X})=3$.  
 ~~$50.5\%$ of trajectories profit.   
 Profit is given in units of USD.}
    \label{fig:V_is_X2_plus_v}
\end{figure}

\begin{figure}[!h]
    \centering
    \includegraphics[scale=.5]{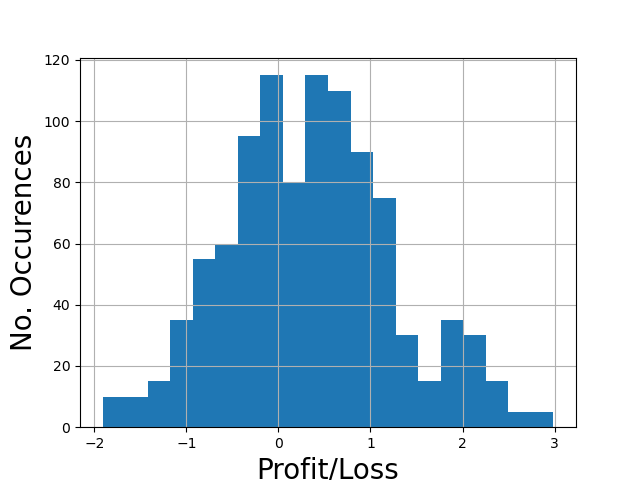}
    \caption{Profit/Loss histograms of 1000 randomly simulated trajectories, under model B and payoff $X^1_{N({\bf X})}= F(X^2_{N({\bf X})})$ with initial investment $V=X_0^1=154.5524$. $\deltahatone=\deltahattwo=0.01$, $\delta=0.011$, and $N({\bf X})=3$.  
    $63.5\%$ of trajectories profit.   
 Profit is given in units of USD.}
    \label{fig:V_is_X1}
\end{figure}

\begin{figure}[!h]
    \centering
    \includegraphics[scale=.5]{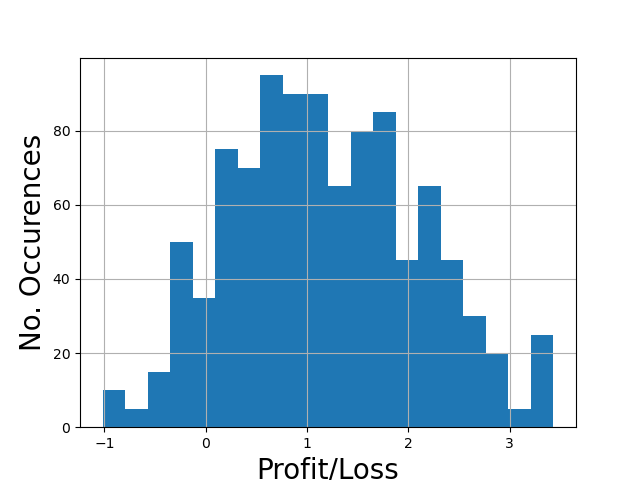}
    \caption{Profit/Loss histograms of 1000 randomly simulated trajectories, under model B and payoff $X^1_{N({\bf X})}= F(X^2_{N({\bf X})})$ with initial investment $V=X_0^1+v=154.5524+1.00$. $\deltahatone=\deltahattwo=0.01$, $\delta=0.011$, and $N({\bf X})=3$.  
 $90.5\%$ of trajectories profit.   
 Profit is given in units of USD.}
    \label{fig:V_is_X1_plus_v}
\end{figure}

\clearpage
\pagebreak[4]

Tables \ref{table:PricesBM}, \ref{table:PricesBM1} and \ref{table:PricesBM2} are obtained for Geometric Brownian motion and using Model A with the following parameters: $\hat{\delta}^1= \hat{\delta}^2= 0.01$, $\delta^{0, A} =0.1, \delta^{1, A} =0.001$
and $N({\bf X})=3$. The ``historical", geometrical Brownian motion simulated charts $x^1$ and $x^2$, were generated with $\mu^1 =0, \sigma^1= 0.01, \mu^2= 0$ and $\sigma^2= 0.02$ respectively.

\begin{table}[!ht]
\centering
 \begin{tabular}{||c ||c|c|l|} 
 \hline
 Brownian Motion Data  &  $~\underline{\sigma}X^2$ & $~\overline{\sigma}X^2$ & $X_0^2$\\
  
  \hline
  Model A & 332.78 & 333.84 &333.37\\ 
 \hline 
 \end{tabular}
\vspace{.1in} 
\caption{Price bounds $\overline{\sigma}X^2$ and $\underline{\sigma}X^2$ where $F(X^1_{N({\bf X})})=X_{N({\bf X})}^2$ and $N({\bf X})=3.$ } 
 \label{table:PricesBM}
\end{table} 
 \begin{table}[!htbp]
\centering
 \begin{tabular}{||c ||c| c|c|c|l|} 
 \hline
 V  &$~\overline{\sigma}X^2$ &$X_0^2$& $X_0^2+0.1$ & $X_0^2-0.1$ &$\underline{\sigma}X^2$ 
 \\
 \hline
 &$100\%$& $61.5\%$& $80\%$   & $52\%$  &$0\%$\\
 \hline
\end{tabular}
\vspace{.1in}
\caption{For each model and initial investment , the percentage of trajectories which profit are recorded. In each case 1000 trajectories $(X_i)_{0\leq i \leq N({\bf X})}$  were simulated with $N({\bf X})=3$.  Notice that, as the initial investment approaches $\overline{\sigma}X^2$ a larger percentage of trajectories are expected to profit, and vice versa for  $\underline{\sigma}X^2$. Note $F(X^1_{(N({\bf X})})=X_{N({\bf X})}^2$, with $X_0^2=333.78$.   
} \label{table:PricesBM1}

\vspace{.1in}
\centering
 \begin{tabular}{||c ||c| c|c|c|l|} 
 \hline
 V  &$~\overline{\sigma}X^2$ &$X_0^2$& $X_0^2+0.1$ & $X_0^2-0.1$ &$\underline{\sigma}X^2$ 
 \\
 \hline
&$0\%$& $60.5\%$& $49.5\%$   & $63\%$  &$100\%$\\
 \hline
\end{tabular}
\vspace{.1in}
\caption{Same remarks as in the Caption to Table \ref{table:PricesBM1} apply to this table which displays Profit and Loss information for the underhedging strategy.} \label{table:PricesBM2}
\end{table}

Clearly, it is possible to experiment with shrinking the trajectory space and then analysing the ensuing profit and loss that results. This could be done by tightening the pruning constraints from Section \ref{pruning1} or, as another alternative, by multiplying the convex hull generated by future tradable coordinates
$(X_{i+1}^1, X_{i+1}^2)$ (i.e. we are conditioning at a given node 
$(X_{i}^1, X_{i}^2)$)  by $1- \epsilon$ for small values of $\epsilon$. 
 
\clearpage
\pagebreak[4]
\section{Arbitrage}  \label{sec:Arbitrage}
During the construction stage, and as a result of dynamic pruning, our models may generate arbitrage opportunities which occur as arbitrage nodes (see Definition \ref{localWithRespectToH} in Appendix \ref{theoreticalFramework}). Without incorporating pruning constraints, our models would never witness an arbitrage opportunity since the set $\{(m^1, m^2):(m^1, m^2, 1, q,  \eta) ~\in N_E\}$ will contain  $0\in\mathbb{R}^2$ (more precisely, this will be a fact, by all practical accounts, as enough historical data is aggregated) and hence it is an arbitrage-free node (see Proposition \ref{characterizationOfNodes} in Appendix \ref{theoreticalFramework}). On the other hand, when pruning constraints are incorporated,  a number of trajectories may be removed from the trajectory set, these correspond to the historically worst case trajectories. The deletion of certain nodes opens up for the possibility of constructing an arbitrage node.
For each node $({\bf X}, i)$ with ${\bf X}_i=(X_i^1,X_i^2,i,T_i,W_i)$ define:
\begin{equation} \nonumber 
E_{({\bf X}, i)}=\{(X_{i+1}^1,X_{i+1}^2),~\exists~~{\bf X}_{i+1}=(X_{i+1}^1,X_{i+1}^2,i+1,T_{i+1},W_{i+1})\in N_A({\bf X}_i)\}, 
\end{equation}
where $N_A({\bf X}_i)$ was introduced in the sentence preceding the display (\ref{admissible}).
We also set 
\begin{equation} \nonumber
\Delta X(E_{({\bf X}, i)}) \equiv \{\Delta _i \tilde{X}= (\tilde{X}^1_{i+1}, \tilde{X}^2_{i+1}) - (X^1_{i}, X^2_{i}): ~~\tilde{{\bf X}} \in E_{({\bf X}, i)} \}. 
\end{equation}
Then, $({\bf X}, i)$ is an arbitrage node of type I if $0 \in [\mathrm{cl}(\mathrm{co}(\Delta X(E_{({\bf X}, i)}))) \setminus \mathrm{ri}(\mathrm{co}(\Delta X(E_{({\bf X}, i)})))]$. Also, $({\bf X}, i)$ is an arbitrage node of type II if $0 \notin \mathrm{cl}(\mathrm{co}(\Delta X(E_{({\bf X}, i)})))$. Here, and elsewhere in the paper, $\mathrm{co}(\cdot)$ and $\mathrm{cl}(\cdot)$ denote the convex hull and   closure of a set, respectively. These facts are presented in Appendix \ref{theoreticalFramework}.

We  emphasize that the only coordinates required in the super/underhedging valuation process are $X^1_i$ and $X^2_i$. Including the other variables, namely $i$, $T_i$ and $W_i$, in our models is solely for the purpose of pruning potential future nodes.  When discussing arbitrage, we refer only to the two asset coordinates and hence the phenomenon is two dimensional.   

A node may only be determined to be arbitrage or arbitrage-free once the adjacent (children) nodes are generated. If a node turns out to be an arbitrage node, all trajectories passing through such a node are terminated at the respective adjacent nodes (i.e. terminated at the earliest possible time). 

Our theoretical framework allows to neglect  Type II nodes  while computing superhedging and underhedging prices. We show in Section \ref{subsec:Ignoring_null_sets}, Appendix \ref{theoreticalFramework}, that Type II nodes are null sets. Moreover, if not ignored, arbitrage nodes of type II will offset the superhedging/underhedging methodology 
as one can easily see that $\overline{\sigma}_jf ({\bf X}) = - \infty$ at a type II node $({\bf X}, j)$.

 In algorithmic terms, rather than artificially deleting Type II nodes (something that it may imply unintended consequences), or adding new nodes to arbitrarily create a no-arbitrage node, our approach is to simply stop the recursive process of generating successive nodes as soon as Type II arbitrage is detected. Nodes with arbitrage are then labelled as such and then ignored by the pricing algorithm. In this way, our trajectory set $\mathcal{X}$ will  be such that trajectories will stop prematurely once arbitrage is detected. In particular, this may lead to different trajectories having different numbers of maximum rebalancing times. 

Dealing with arbitrage nodes of type II is, in practical terms,  the main part of the story but it does not cover all possible cases. The key underlying
property that needs to hold at a given node $({\bf X}, i)$, in order to evaluate superhedging prices, is property $(L_{({\bf X}, i)})$. This property is introduced  in Definition \ref{propertyL} and discussed in detail afterwards; the property can fail if $(\tilde{{\bf X}}, i+1)$, with $\tilde{{\bf X}} \in \mathcal{X}_{({\bf X}, i)}$, is a type II arbitrage node and 
$\mathcal{X}_{({\bf X}, i)} \setminus \{\tilde{{\bf X}}\}$, in turn, is also a type II node. That is,
the removal of a trajectory containing a type II child node makes the original given node a type II node itself. This phenomena is illustrated in Figures \ref{fig:arbitrage demo} and \ref{fig:arbitrage_plot_without_arbitrage}.

Arbitrage nodes of type I are less likely to occur during our trajectory construction
and they do not correspond (entirely) to null events. The way they are handled is described in Section \ref{subsec:Ignoring_null_sets}.

To summarize the results from Section \ref{subsec:Ignoring_null_sets}:
we will evaluate superhedging and underheding prices at nodes where the property $(L_{({\bf X}, i)})$ holds and the latter will be required (as a general property in our models)
to hold a.e. Then by Theorem \ref{thm:sigma_is_correct}
superhedging prices can be computed with a simple portfolio but will satisfy that the super/under hedging  property will be uphold only a.e. as opposed for all trajectories.

\begin{figure}[!htb]
    \centering
    \includegraphics[scale=.80]{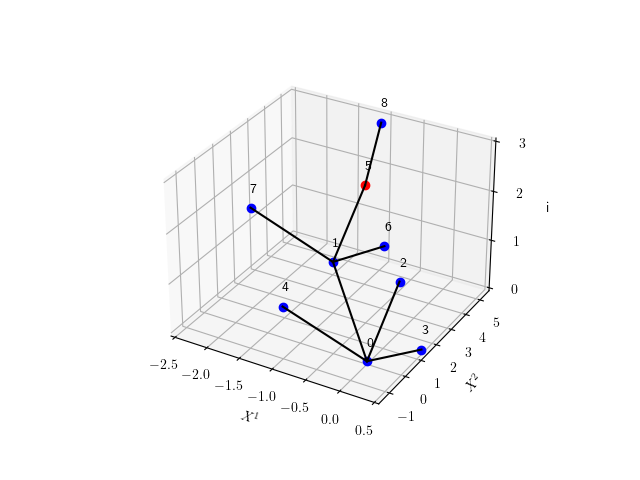}
    \caption{Depiction of a trajectory set $\mathcal{X}$, with node 5 being an arbitrage node of type II. At node 0 the prices are given as  $\overline{\sigma}F(X)=2, \underline{\sigma}F(X)=-2$, where $F(X)=X_{N({\bf X})}^2, N({\bf X})=3.$  In this case, the removal/ignoring of node $5$ while computing  at node $1$ will make the latter, in turn, a type II node. In particular property $(L_{(\hat{{\bf X}}, 1)})$ does not hold (where $(\hat{{\bf X}}, 1)$ denotes node $1$). See further explanations in the caption to  Figure \ref{fig:arbitrage_plot_without_arbitrage}.}
    \label{fig:arbitrage demo}
\end{figure}
\begin{figure}[!htb]
    \centering
    \includegraphics[scale=.90]{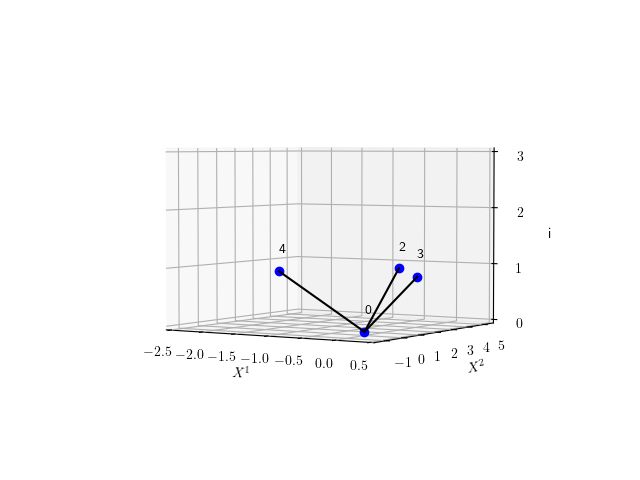}
    \caption{
    Let $\hat{\mathcal{X}}$ denote the set of  all the trajectories passing through node $1$ in Figure \ref{fig:arbitrage demo}. It is possible to see that
    $\hat{\mathcal{X}}$ is a null set and that the property $(L)-a.e.$ holds (as defined in Appendix \ref{subsec:Ignoring_null_sets}). We then know, also by referring to Appendix \ref{subsec:Ignoring_null_sets}, that undergedging and superhedging prices will coincide in the two trajectory sets 
    $\mathcal{X} \setminus \hat{\mathcal{X}}$ and  $\mathcal{X}$, namely: $\overline{\sigma}F(X)=2, \underline{\sigma}F(X)=-2$, where $F(X)=X_{N({\bf X})}^2, N({\bf X})=1$ with ${\bf X} \in  \mathcal{X} \setminus \hat{\mathcal{X}}$.}
    \label{fig:arbitrage_plot_without_arbitrage}
\end{figure}

\begin{figure}[!h]
    \centering
    \includegraphics[scale=.90]{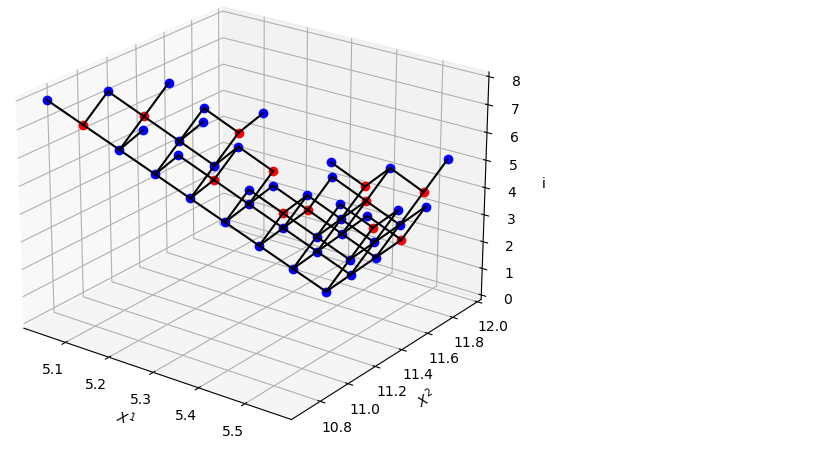}
    \caption{An example of a trajectory set as a graph (showing only the first three coordinates $(X_i^1, X_i^2,i)$. The graph is directed in the following sense:   nodes with coordinate $i$ are connected to nodes with coordinate $i+1$ but not vice versa. The figure demonstrates the (potential) early/premature  termination of some trajectories, as well as the creation of Type II arbitrage nodes as a result of applying pruning constraints. The trajectory set is constructed under model B  with $\delta=0.1, ~N({\bf X})=8$. Arbitrage nodes are denoted in red and are always of type II while no arbitrage nodes are denoted in blue.  Notice the arbitrage nodes have either one or two children. When the number of spawned children is small, it is quite likely that a trajectory may terminate prematurely due to all nodes being pruned or due to too many arbitrage nodes. A spawning size of  $3$ is the smallest size that allows a trajectory set to unfold, as any smaller set automatically leads to arbitrage. The graph contains 51 nodes and 78 edges. 
}
    \label{fig:my_label}
\end{figure}

\clearpage
\pagebreak[4]
\subsection{Small Arbitrage and Dubin's Two Dimensional Cone Crossings Inequality} \label{smallArbitrages}

Here we describe  how the trajectorial theoretical framework can be used to further prune along while constructing a trajectory set. The motivation for the proposed pruning in this section is quite distinct from the worst case approach used elsewhere in the paper. The proposed methodology is an extension to treating arbitrage opportunities as null events.

We define a {\it small arbitrage} if for any $\delta >0$ there exist a function $f_{\delta}$ with domain $\mathcal{X}$ such that $0 \leq f_{\delta} \leq c$, where $c \geq 1$ is the maximum of $f_{\delta}$ over $\mathcal{X}$, and $\overline{\sigma}f_{\delta} \leq \delta$.

We now argue, informally, that if $f_{\delta} = c~\textbf{1}_A$, where $A \subseteq \mathcal{X}$,
then it will be unlikely that trajectories in $A$ will unfold in actual markets (under the assumption that unfolding charts will be contained in $\mathcal{X}$). The simple argument is that $\delta$ can be chosen arbitrarily small and the potential maximum  payoff of $f_{\delta}$ remains equal to  $c \geq 1$. Given the meaning of $\overline{\sigma}$, we can then set up a portfolio with at most initial value $\delta$ that will superhedge $f_{\delta}$ on $\mathcal{X}$. We see then $f_{\delta}$
as a lottery which price can be made arbitrarily small by reducing the value of $\delta$ while its maximum  payoff does not decrease as a function of its price. It is then reasonable, for $\delta$ small, to infer that such lottery will not be available, an upshot is that the subset $A$ is unlikely to occur. The cut-off value of $\delta$ is dependent on the modeller/investor and, hence, can be used to exchange uncertainty for reward. Notice that the case of a null
event $\overline{\sigma}({\bf 1}_A)= \overline{I}({\bf 1}_A)=0$ is a special 
case of a small arbitrage. In particular, given that arbitrage opportunities are null in our trajectorial setting, we see that arbitrage opportunities will be small arbitrages.  On the other hand,  an small arbitrage is not an arbitrage as the initial investment $\delta >0$ may not be recuperated but, being a quantity that is investor-dependent small, it may not be a deterrent for investment which, in turn, will prompt potential large losses
for sellers of such an option (or traders taking the dual side of the trade).

Another way to reach a similar conclusion is to note that $\overline{\sigma}({\bf 1_A})$ upperbounds 
$Q(A)$ for any martingale measure on $\mathcal{X}$; supposing such pricing measure $Q$ equivalent to a measure $P$ we may expect $P(A)$ to be small (as we can make $\delta \rightarrow 0$). As this reasoning holds for any such measure $P$, one is then lead to infer the fact that $A$ will be unlikely for any such potential physical measure $P$.

We develop now an example of a small arbitrage; to this end we need the notion of a {\it trajectorial supermartingale} (studied in detail in \cite{bender3}):
this refers to a sequence of non-anticipative functions $f_j: \mathcal{X} \rightarrow \mathbb{R}$ satisfying
\begin{equation} \label{supermartingale}
\overline{\sigma}_j f_{j+1} \leq f_j ~~a.e.
\end{equation}
where the notion on a.e. is non-probabilistic and has been introduced in Definition \ref{nullObjects}.
 The notion of non-anticipativity was introduced in Section \ref{portfolioSets} and means $f_n({\bf X})= f_j({\bf X}_0, \ldots, {\bf X}_j)$. A trajectorial supermartingale becomes a {\it trajectorial martingale} if the inequality  in (\ref{supermartingale}) is replaced by equality  (a.e. and required for all $j$).

Dubin's classical upcrossing inequality counts the upcrosses of a non-negative supermartingale $\{f_k\}$ through a given band $[a, b]$ (\cite{neveu}), $0 \leq a < b$, and it can be extended to the case of a trajectorial supermartingale. This result will be reported elsewhere (for a preliminary version see \cite{konrad}). The above definition of small arbitrage and ensuing discussion can be applied to such non-probabilistic extension of Dubin's inequality. Nonetheless, and in order to provide a version of Dubin's inequality closer to the two dimensional setting of the paper, we describe below a novel, alternative, version to the classical version of Dubin's inequality, a generalization of sorts, involving two sequences jointly upcrossing through a given cone.

\subsubsection{Anchored Cone-Crossings} 

The non-anticipativity property, assumed below, is required so that when defining the counting times they turn out to be stopping times (in a trajectorial sense as introduced in \cite{bender3} or \cite{konrad}). The stopping time property is needed in the proof of Theorem \ref{dubinForRationsAndFiniteTime} below as it relies on a trajectorial version of  Doob's optional sampling theorem (as presented in \cite{konrad}, for example).

\vspace{.1in}
\begin{definition} \label{anchoredUpcrossings}
Let $\{f^2_i\}_{0 \leq i \leq N}$ be a  non-negative and non-anticipative sequence of functions and 
$\{f^1_i\}_{0 \leq i \leq N}$ a positive sequence of functions that is non-anticipative (in each case functions defined on $\mathcal{X}$).  
For $ 0 \leq \alpha < \beta$ we define upcrossing times recursively by setting $\tau_0=0$  and  $\tau_k \leq \rho_k \leq \tau_{k+1}$ to be upcrossing times as follows
\begin{equation}  \label{downCone}
\frac{f^2_{\tau_{k}}}{f^1_{\rho_k}} \leq \alpha
\end{equation}
and
\begin{equation}  \label{upCone}
\frac{f^2_{\tau_{k+1}}}{f^1_{\rho_k}} \geq \beta.
\end{equation} 
\end{definition}

\vspace{.1in}
In words: if we consider $\alpha < 1 < \beta$, we reason as follows, given $f^2_{\tau_k}$, we look for $\rho_k \geq \tau_k$ such that $f^1_{\rho_k} \geq f^2_{\tau_k}/\alpha$ i.e. $f^1$ moves up relative to $f^2$. Then we look for $\tau_{k+1} \geq \rho_k$ such that $f^2_{\tau_{k+1}} \geq f^1_{\rho_k}~\beta$ i.e. $f^2$ moves up relative to $f^1$. Notice that using the same $\rho_k$
in both expressions above plays a role of anchoring the consecutive times as illustrated in the Figure \ref{anchoredConeCrossings} below which shows that cone-crossings are vertical as the $x$-coordinate is the same.

The counting relationships (\ref{upCone}) and (\ref{downCone}) imply
\begin{equation}  \label{moreCounting}
\frac{f^2_{\tau_{k+1}}}{f^2_{\tau_k}} \geq \frac{\beta}{\alpha} = \lambda.
\end{equation}
The fact that the upper bound that appears in Dubin's inequality only depends on $\alpha/\beta$ introduces a rather arbitrary sequence $\{f^1_n\}$ in each counting and this degree of freedom is exploited by the ratio formulation. $f^1_n$ only needs to be adapted (i.e. non-anticipative) as it is used to define
the counting times $\rho_k$ which in the proof will need to be stopping times (a property that will follow from the $f^1_n$ being non-anticipative).


\vspace{.2in}

\begin{figure}[!htb]
\centering
\begin{tikzpicture} 
\begin{axis}[
    axis lines = middle,
    xlabel = {$f^1_j$},
    ylabel = {$f^2_j$},
    xmin = 0, xmax = 6,
    ymin = 0, ymax = 5,
    grid = both,
    width=10cm,
    height=8cm,
    xtick=\empty,  
    ytick=\empty,  
    axis line style={-}, 
    enlarge x limits=true, 
    enlarge y limits=true
]

\addplot[domain=0:5, samples=2, red, dashed] {0.5*x}; 
\addplot[domain=0:5, samples=2, green, dashed] {1.5*x}; 

\addplot[only marks, mark=*, purple] coordinates {(0.9, 0.3)};
\addplot[only marks, mark=*, purple] coordinates {(3.6, 1.7)};

\addplot[only marks, mark=*, orange] coordinates {(0.9, 1.7)};
\addplot[only marks, mark=*, orange] coordinates {(3.6, 5.5)}; 

\node at (axis cs:2.4,0.3) [anchor=east] {$(f^1_{\rho_k}, f^2_{\tau_k})$};
\node at (axis cs:5.8,1.8) [anchor=east] {$(f^1_{\rho_{k+1}}, f^2_{\tau_{k+1}})$};

\node at (axis cs:1.4,1.7) [anchor=south] {$(f^1_{\rho_k}, f^2_{\tau_{k+1}})$};
\node at (axis cs:4.6,4.8) [anchor=south] {$(f^1_{\rho_{k+1}}, f^2_{\tau_{k+2}})$};

\end{axis}
\end{tikzpicture}
\caption{\label{anchoredConeCrossings} Illustration of Anchored Cone-Crossings} 
\end{figure}
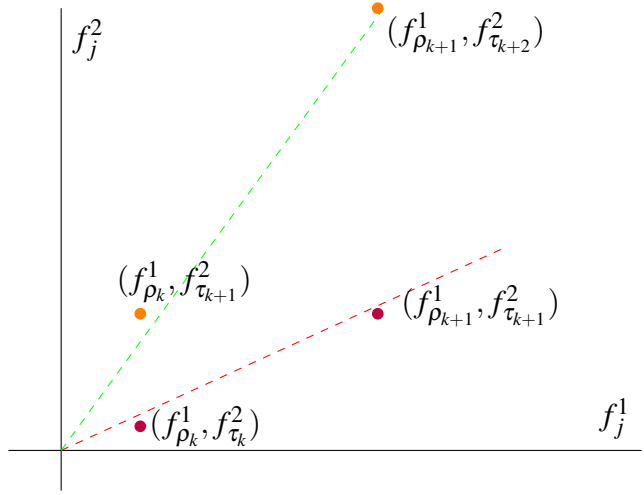

We present the next result without proof.
\begin{theorem}[Dubin's inequality for ratios. Counting as per Definition \ref{anchoredUpcrossings}] \label{dubinForRationsAndFiniteTime}
Let $\{f^2_i\}_{0 \leq i \leq N}$ be a, real-valued,  positive trajectorial supermartingale (as per (\ref{supermartingale}))  and 
$\{f^1_i\}_{0 \leq i \leq N}$ a positive sequence of non-anticipative functions. We require  $f^2_j({\bf X}) \geq \delta^2 >0$,  for all ${\bf X} \in \mathcal{X}$ and all $j$.
For $ 0 \leq \alpha < \beta$, consider the definitions of $\rho_k, \tau_k$ as given in Definition \ref{anchoredUpcrossings} and taking the value $N+1$ whenever the optimal conditioning sets are empty. Furthermore, assume
that $\mathcal{X}$ is  such that $(L)-a.e.$ holds as per Definition \ref{La.e.}; then,
for any $k\ge 0$ and  $0 \leq \alpha < \beta$:
\begin{equation} \nonumber 
\overline{\sigma}\left({\bf 1}_{\{\tau_{k+1} < N+1\}}\right) \leq ~~ \overline{\sigma} \left(\frac{\alpha}{\beta}~ {\bf 1}_{\{\rho_{k}  < N+1\}}
 \right)  \leq \left(\frac{\alpha}{\beta} \right)~~ \overline{\sigma} \left(~ {\bf 1}_{\{\rho_{k}  < N+1\}}
 \right).
\end{equation}
Therefore, 
\begin{equation}
\overline{\sigma}({\bf 1}_{\{\tau_{k+1} < N+1\}}) \leq \left(\frac{\alpha}{\beta}\right)~~ \overline{\sigma}({\bf 1}_{\{\tau_{k}  < N+1\}}),
\end{equation}
and so
\begin{equation} \label{couldBeCheckedEmpirically}
\overline{\sigma}({\bf 1}_{\{\tau_{k+1}  < N+1\}}) \leq
\left(\frac{\alpha}{\beta} \right)^{k+1} ~\overline{\sigma}({\bf 1}_{\{\rho_{0}  < N+1\}}) \leq
\left(\frac{1}{\lambda}\right)^{k+1}.
\end{equation}
where $ \lambda = \frac{\beta}{\alpha}$.
\end{theorem}

Theorem \ref{dubinForRationsAndFiniteTime} relates to the previous discussion on small arbitrage by noticing that $A_{k+1, \lambda} \subseteq  A_{k, \lambda}$ where $A_{k, \lambda} \equiv \{{\bf X} \in \mathcal{X}: \tau_k({\bf X}) < N+1\}$ for a fixed but arbitrary parameter $1 <\lambda \equiv \beta/\alpha < \infty$. By increasing $k$ we can achieve $\overline{\sigma} f_{\delta}  \leq \delta \equiv c~(1/\lambda)^k$ 
where $f_{\delta} \equiv c {\bf 1}_{A_k \lambda}$. As we explain below, trajectories belonging to $A_{k, \lambda}$ represent pairs of price sequences
that upcross a given cone in $2$-dimensions. Theorem \ref{dubinForRationsAndFiniteTime} proving $\overline{\sigma} {\bf 1}_{A_{k, \lambda}} \leq (\frac{1}{\lambda})^k$ is a novel  $2$-dimensional analogue of the classical upcrossing inequality of Dubin (\cite{neveu}).

To connect Theorem \ref{dubinForRationsAndFiniteTime} with our $2$-dimensional trajectory construction we let $f^2_j = X^2_j$ and $f^1_j = X^1_j$ for the case that we only trade with $X^2_j$ in order to superhedge $X^1_N$. This will guarantee that $f^2_j$ is a  trajectorial martingale and hence a trajectorial supermartingale (as required in Theorem \ref{dubinForRationsAndFiniteTime}). That is, we formally rely on Theorem \ref{dubinForRationsAndFiniteTime} as a $1$-dimensional setting
by considering the trajectory space consisting of trajectories $X^2_j$; this choice  allows us to have access to \cite{bender3} for sufficient conditions providing the validity of $(L)-a.e.$ (which is required  in Theorem \ref{dubinForRationsAndFiniteTime}). Interestingly, the theorem does not require a supermartingale property for $f^1_j= X^1_j$ and this fact allows to apply the results of the theorem to our $2$-dimensional construction. One could formulate an alternative counting setting where 
both $X^k_j$, $k=1,2$, are required to be supermartingales, such a formulation will place us in a $2$-dimensional setting which, in turn, would require dealing with $(L)-a.e.$ in $2$-dimensions as well.

It then follows that Theorem \ref{dubinForRationsAndFiniteTime} can be used to prune
the two dimensional trajectories $(X^1_j, X^2_j)$ by stopping their recursive construction whenever 
$(\frac{1}{\lambda})^k$ is considered to be sufficiently small and this can be done for each cone $0 < \alpha < \beta$ with $\lambda \equiv \frac{\beta}{\alpha}$. As explained above this type of pruning is suggested apriori by the theoretical framework which allows to prove a novel version of Dubin's inequality (namely, Theorem \ref{dubinForRationsAndFiniteTime}).


\section{Discussion}  \label{sec:discussion}

Standard stochastic price modelling attempts to capture
observable features of time series under the assumption of no-arbitrage. Using stochastic language, we could explain our operationally-based approach by associating, to a given set of agents, a class of stopping times that they will use to sample a given stochastic process model. Therefore, agents that trade following their own special
investment rules, will only face special features of such stochastic process. This reasoning suggests to directly construct models that do reflect prices faced by a relevant set of traders/agents. Following this idea,
our paper proposes a data-based, systematic, modelling construction that reflects trading behaviours of a class of agents. Specifically, we propose a class of agents that rebalance their portfolios
after $2$-dimensional $\delta$-escape price movements away from its present value.
Once those historically observable samples are acquired from data, we then use a worst case approach that is supported by a non-probabilistic theory to create a trajectorial model. The mathematical theory that we rely upon is built from financial insights and relies on the notion of superhedging.  Our approach allows to incorporate arbitrage opportunities as null sets without any resort to a probability measure.
The modelling methodology permits agents to see investing more akin to a casino by allowing them to gauge risk-reward possibilities which are directly associated to their portfolio rebalancing methodology.

\begin{appendices}

\section{Appendix A. Theoretical Framework} \label{theoreticalFramework}

This appendix presents theory that supports and justifies the superhedging framework used in the paper. We only provide detailed statements and proofs
for results with a direct relevance to the paper and provide references for background results and their proofs. As we have already anticipated, at an early point in our explanations we switch to a $1$-dimensional setting as that is the dimension required for the superhedging algorithm that we rely upon. We will be referring to some results available in 
\cite{bender3}, the setting of that reference is for $d=1$ and positive portfolios
(appearing in the definition of the superhedging operators) are defined by $f_m = \liminf_{n \rightarrow \infty} \Pi^{V_m, H^m}_{j, n}$. That is, the definition involves a $\liminf$ and it represents a more general setting that the one of the present paper where 
 $f_m = \Pi^{V_m, H^m}_{j, n_m}$ which is a particular case of the $\liminf$ framework by taking $H^m_k=0$ for all $k \geq n_m$.
 
\subsection{Types of Nodes and Arbitrage} \label{nodeTypes}

Here we introduce definitions for the different types of nodes, their geometrical characterizations and how these notions relate to the 
usual definition of no-arbitrage portfolio and absence of arbitrage opportunities.

\noindent
The inner product in $\mathbb{R}^2$ is denoted by $h \cdot Y$; for convenience, the definitions and results are presented only for $\mathbb{R}^2$ but are
still available for any $d \geq 1$, where $d$ refers to the number of traded assets in the model, 
 (in fact, we use 
the notions and ensuing properties for the case  $d=1$ case as well).

Below, and elsewhere in the paper, a node $({\bf X}, k)$ is a shorthand notation for the set $\mathcal{X}_{({\bf X}, k)}$.

\begin{definition}[Arbitrage and $0$-Neutral Nodes] \label{localWithRespectToH}
  Given a (multidimensional) trajectory set $\mathcal{X}$ and  a node $({\bf X}, k)$, $k \ge 0$:
\begin{enumerate}
  \item $({\bf X}, k)$ is called an {\it arbitrage-free node} if for any $h \in \mathbb{R}^2$
  \[
  [h \cdot \Delta_k X'=0 \quad \forall X' \in \mathcal{X}_{({\bf X},k)}]~ ~~\mbox{or}~~~ [\inf_{X' \in \mathcal{X}_{({\bf X},k)}} h \cdot \Delta_k X' < 0].
  \]
 \item $({\bf X}, k)$ is called a  $0$-neutral node if for any $h \in \mathbb{R}^2$
  \[
  \inf_{X' \in \mathcal{X}_{({\bf X},k)}} h \cdot \Delta_k X' \le 0.
  \]
\end{enumerate}
$({\bf X}, k)$ is called an arbitrage node if it is not an arbitrage-free node.
\end{definition}

Relying on Proposition 3.5 of \cite{degano2} we obtain the following result where $\mbox{ri}(A)$, $\mbox{co}(A)$ and $\mbox{cl}(A)$  refer to relative interior, convex hull and closure, respectively, of a set $A \subseteq \mathbb{R}^2$.

\begin{proposition}  \label{characterizationOfNodes}
Given a trajectory set $\mathcal{X}$, consider a
node $({\bf X},k)$.

\begin{itemize}
\item  $({\bf X},k)$ is an arbitrage-free node if and only if 
\begin{equation} \nonumber 
0 \in \mathrm{ri}\left(\mathrm{co}\left(\Delta X(\mathcal{X}_{({\bf X},k)})\right)\right).
\end{equation}

\item  $({\bf X},k)$ is a $0$-neutral node if and only if
\begin{equation} \nonumber 
0 \in \mathrm{cl}\left(\mathrm{co}\left(\Delta X(\mathcal{X}_{({\bf X},k)})\right)\right).
\end{equation}
\end{itemize}
We relied on the notation introduced in (\ref{eqn:conjdelta}) for $\Delta X(\mathcal{X}_{({\bf X},k)})$.
\end{proposition}

$({\bf X},k)$ will be called a {\it type I arbitrage node} if it is a $0$-neutral node
but not an arbitrage-free node, i.e. whenever $0$ is in the boundary of the closure of the convex hull. We will call $({\bf X},k)$ a {\it type II arbitrage node} if it is neither a type I arbitrage node nor an arbitrage-free node. 
It then follows that a node $({\bf X},k)$ is a type II node if:
\begin{equation}\nonumber
   0 \notin \mathrm{cl}\left(\mathrm{co}\left(\Delta X(\mathcal{X}_{({\bf X},k)})\right)\right). 
\end{equation}
Notice that whenever $\mathcal{X}$ is a finite set,  there are no type II arbitrage nodes that are also 0-neutral. This then, will be the case when we restrict to computer implementations of the models. $({\bf X},k)$ is a type I arbitrage node if
\begin{equation}\nonumber
  0 \in \mathrm{cl}\left(\mathrm{co}\left(\Delta X(\mathcal{X}_{({\bf X},k)})\right)\right) \quad \text{and}\quad 0 \notin \mathrm{ri}\left(\mathrm{co}\left(\Delta X(\mathcal{X}_{({\bf X},k)})\right)\right).
\end{equation}

In financial terms, an arbitrage-node allows an investor 
to place a trade at that node without the possibility of losing any money.
The special case of an arbitrage-node of type I is when there is a possibility of earning nothing. This is a rare 
case
and unlikely to appear in practice, nonetheless, it is an interesting case in that it 
allows for a pricing methodology while allowing for (some) arbitrage opportunities as well (see \cite{ferrando}).
A trade at an arbitrage-free node  will always have the possibility to lose money
(or to earn nothing for all possible cases).

A common modelling assumption is not to allow for investors that are able to generate a profit in a transaction without any risk/possibility of losing money.  Such an investment opportunity is called an arbitrage opportunity.

\begin{definition}[Arbitrage opportunity] \label{ArbitrageDefinition}
Given a trajectory set $\mathcal{X}$  a portfolio $H=({\bf H}_i)_{i \geq 0}$ is called an arbitrage opportunity if
\begin{itemize}
\item[i)] $\forall ~\se \in \mathcal{X}$,  $V^{0, H}_{0, N}(\se) \geq 0$;
\item[ii)] $\exists ~\se^{\ast} \in \Se$ such that $V^{0, H}_{0, N}(\se^{\ast}) > 0,$
\end{itemize}
for some trajectory-dependent index $N= N({\bf X})$.
We say that $\mathcal{X}$ is \emph{arbitrage-free} if there is no arbitrage portfolio $H$.
\end{definition}

According to Theorem 3.9 and Proposition 3.10  in \cite{degano2} all nodes need to be no-arbitrage nodes so that there are no arbitrage portfolios.

\vspace{.1in}
As we already discussed in Section \ref{detectingNullEvents}, 
our superheding norm $\overline{I}$ defines and detects null sets and the latter are interpreted as unlikely events. One then has the possibility to allow
for arbitrage opportunities by weakening condition $i)$ in Definition \ref{ArbitrageDefinition} to only require $V^{0, H}_{0, N}(\se) \geq 0 ~~a.e.$
That is $V^{0, H}_{0, N}(\se) \geq 0$ may not hold on a $\overline{I}$-null set (naturally, we would also require that the strict inequality in $ii)$ above holds on a
set that is not null)
this extension is of course in line with the stochastic approach but the difference is that we rely on the financial definition of null events given by 
$\overline{I}$ (as contrasted to a measure based notion of a.e.). It is possible to see that a $\overline{I}$-null set will be also a null set with respect to any probability
measure that makes the coordinate maps $T_k(X^n) \equiv X_k$, where $n \in \{1, 2\}$, into (trajectorial) martingale processes
(see \cite{bender1}).

\subsection{Condition $(L_{({\bf X},j)})$}
The following definition will be a minimal necessary hypothesis required in several results presented later in this appendix  (Proposition \ref{basicInequalities}, Theorem \ref{thm:sigma_is_correct} and Corollary \ref{1DimensionalInBetweenBounds}).

\begin{definition}[{\bf Property} $(L_{({\bf X},j)})$~] \label{propertyL}
Fix: $({\bf X},j)$,  $f_m =  \Pi^{V^m, H^m}_{j, n_m}$ with $\Pi^{V^m, H^m}_{j, n} \in \mathcal{E}_{({\bf X}, j)}^+~\mbox{for all}~ n \geq j,~~ m \geq 1$ and
$f_0 \in \mathcal{E}_{({\bf X},j)}$. Define property $(L_{({\bf X},j)})$ by
\begin{equation}  \nonumber
 [~0 \leq \sum_{m \geq 0} f_m\;\;\mbox{on}\; \mathcal{X}_{(X,j)} \implies  0 \leq \sum_{m \geq 0} V^m~].
\end{equation}
\end{definition}

Property $(L_{({\bf X},j)})$ is fully discussed in \cite{bender3} and the reader is referred to that reference to appreciate why $(L_{({\bf X},j)})$ is a cornerstone property. On the other hand, the latter reference is in a $1$-dimensional setting while Definition \ref{propertyL} is meant to cover the multidimensional case (see clarifying comments in Remark \ref{validityOfL} below). 
The $2$-dimensional version of $(L_{({\bf X},j)})$ is required in the next proposition.
\begin{proposition} \label{basicInequalities}
Given a $2$-dimensional trajectory set $\mathcal{X}$ with traded coordinates
$X_i= (X^1_i, X^2_i)$, $0 \leq i \leq N({\bf X})$. Fix a node $({\bf X}, j)$ and assume that $(L_{({\bf X},j)})$ holds. If we define $F({\bf X}) \equiv X^2_{N({\bf X})}$ then:
\begin{equation}  \label{eqPriceInBetweenBounds}
\underline{\sigma}_j F({\bf X}) \leq X^2_{j} \leq \overline{\sigma}_j F({\bf X}).
\end{equation}
\end{proposition}
\begin{proof}
Consider $ X^2_{N({\bf X})} \leq \sum_{m \geq 0} f_m$ on $\mathcal{X}_{({\bf X}, j)}$ where $f_m \equiv \Pi^{V^m, H^m}_{j, n_m}$. Then $0 \leq \sum_{m \geq 0} f_m- X^2_{j}- \sum_{k=j}^{N({\bf X})-1} (X^2_{k+1}- X^2_k)$. From $(L_{({\bf X},j)})$ we then obtain $0 \leq - X^2_j + \sum_{m \geq 0} V^m$ from where it follows that $X^2_j \leq \overline{\sigma}_j F({\bf X})$.
The inequality $\underline{\sigma}_j F({\bf X}) \leq X^2_{j}$ is obtained from the same argument but now applied to $-X^2_{N({\bf X})}$. 
\end{proof}

\subsection{Computational Version of Superhedging Prices}

For this section, the setting for Theorem \ref{thm:sigma_is_correct} below  is $1$-dimensional; in particular, we will use the notation $X \in \mathcal{X}$ for $1$-dimensional trajectories $X=(X_i)_{i \geq 0}$ (i.e. $X_i = X_i^1$).

Theorem \ref{aEComputationWithSimplePortfolios} in Section \ref{sec:computation} shows   that under some conditions, and in order to evaluate the superhedging price $\overline{\sigma}f$,
 we could resort to use simple portfolios but, in that case, would need to require the superhedging to hold only a.e. As already anticipated, we reformulate that result as Theorem \ref{thm:sigma_is_correct} below and take the opportunity to present
 the result for $\overline{\sigma}_j f$ (i.e. a conditional version). Implicit in the definition of $\overline{\sigma}_jf$ is the existence of a generalized portfolio (i.e of the form $\sum_{m \geq 0} f_m$) that superhedges $f$ at all ${\bf X} \in \mathcal{X}$.  Theorem \ref{thm:sigma_is_correct} replaces the generalized portfolio by constructing a
 simple portfolio that instead superhedges a.e. and specifying a concrete null set $\mathcal{N}$ (introduced below). Theorem \ref{thm:sigma_is_correct} thus opens the way to an algorithm to evaluate $\overline{\sigma}_j f$ that is presented in Section \ref{sec:Pricing}.
 
The following set collects the trajectories that pass through arbitrage nodes while
not staying constant at such a node.
$$
\mathcal{N} \equiv \{X \in \mathcal{X}: \exists~~ j\ge 0 ~~~~~ \mbox{s.t.} \; (X,j)\; \mbox{is an arbitrage node and}\; X_{j+1} \neq X_j\}.
$$
We remark that $\mathcal{N}$ is a null set (as per Lemma A.3 in \cite{bender3}) and that  Theorem \ref{thm:sigma_is_correct} requires 
the following  property.
\begin{definition}[$(L)-a.e.$] \label{La.e.}
We will say that $(L)-a.e.$ holds if the two following two conditions are both valid,
 \begin{itemize} \item  $(L_{(\hat{X}, k)})$ holds for all $k \geq 1$ and $\hat{X} \in \mathcal{X} \setminus \mathcal{N}$,
\item  $(L_{(X, 0)})$ holds.
\end{itemize} 
\end{definition}
Sufficient conditions for the validity of $(L)-a.e.$ are presented in \cite{bender3} (see Corollary C.3 in that reference).

The following result is a more specific version of Theorem \ref{aEComputationWithSimplePortfolios} in that the null sets appearing in the latter are replaced below by the concrete null set $\mathcal{N}$. It is possible to see that the assumption $(L_{(X,0)})$
implies that $\mathcal{X} \neq \mathcal{N}$ (i.e. this additional assumption avoids the trivial case).  $\overline{\sigma}_j f (X)$ appearing in the left hand side of (\ref{removingInfiniteSum}) below  is a  special case of Definition \ref{cond_integ_def} (corresponding to $d=1$).

\begin{theorem}\label{thm:sigma_is_correct}
 Let $f: \mathcal{X} \rightarrow \mathbb{R}$ to have finite maturity $n_f\in \mathbb{N}$, i.e., $f(X)=f(X_0,\ldots, X_{n_f})$ for every $X \in \Se$.
Assume the following: \\
 $i)$ if  $\hat{X} \in \mathcal{X} \setminus \mathcal{N}$ then
$(L_{(\hat{X}, k)})$ holds for all nodes $(\hat{X}, k),~k \geq 1$, $ii)$ $(L_{(X,0)})$ holds. Then, the following equality holds for any given $X \in \mathcal{X}$ and $j$ satisfying $0 \leq j \leq n_f$:
\begin{equation} \nonumber
\overline{\sigma}_j f (X)=  \inf \{V \in \mathbb{R}: ~~\exists~~ (H_j)_{j=0,\ldots, n_f-1} ~\textnormal{non-anticipative such that} 
\end{equation}
\begin{equation} \label{removingInfiniteSum}
~~ f(\hat{X}) \leq V(X_0, \ldots, X_j) +\sum_{k=j}^{n_f-1} H_k(\hat{X})\Delta_k \hat{X}) \textnormal{ for all}~\hat{X} \in \mathcal{X}_{(X,j)} \setminus \mathcal{N} \} \equiv \overline{V}_j f(X).
\end{equation}
\end{theorem}
\begin{proof}
Consider an arbitrary node $(X, j)$, $0 \leq j \leq n_f$ fixed for the proof.
We first  establish that the right hand side of (\ref{removingInfiniteSum}) is bounded by  $\overline \sigma_j f(X)$ for all $X \in \mathcal{X}$. Without loss of generality we may then assume that $ \overline \sigma_j f(X) < \infty $.
Consider $f_m=~\Pi_{j,n_m}^{V^m,H^m},~~\Pi_{j,n}^{V^m,H^m} \in \mathcal{E}^+_{(X, j)}$ for all $n \geq j$ and $ m\ge 1, \sum_{m=1}^{\infty} V^m(X) < \infty$, ~\mbox{and}~~	$f_0=~\Pi_{j,n_0}^{V^0,H^0} \in \mathcal{E}_{(X, j)}$  
satisfying
\begin{equation} \nonumber 
		f\le \sum\limits_{m=0}^\infty f_m,\;\;~~\mbox{on}~~\mathcal{X}_{(X,j)}.
\end{equation}
From our assumptions we know  that whenever $\hat{X} \in \mathcal{N}^C$, it follows that for all nodes $(\hat{X}, k)$ $(L_{(\hat{X}, k)})$ holds. Therefore, the Finite Maturity  Lemma 4.3 from \cite{bender3}  implies
\begin{equation}\label{toProve1}
f \le \sum_{m=0}^\infty \Pi_{j,n_f}^{V^m,H^m}\;\;\mbox{on}~~ \mathcal{X}_{(X, j)} \setminus \mathcal{N}.
\end{equation}
Whenever $\hat{X} \in  \mathcal{X}_{(X, j)} \setminus \mathcal{N}$ 
we know that $(\hat{X}, k)$, $k \geq 0$, is an up-down node or a node satisfying $\Delta_k \hat{X} =0$, these facts and $ \overline \sigma_j f(X) < \infty $, the Aggregation  Lemma 4.4 from \cite{bender3} applies and allows to rewrite 
(\ref{toProve1}) as follows
\begin{equation}\label{afterAggregation}
f \le \Pi_{j,n_f}^{V,H}\;\;\mbox{on}~~ \mathcal{X}_{(X, j)} \setminus \mathcal{N},
\end{equation}
where $V(X) \equiv \sum_{m=0}^{\infty} V^m(X)$ and, for $\tilde{X} \in \mathcal{X}$, $H_k(\tilde{X})= \sum_{m=0}^{\infty} H^m_k(\tilde{X})$ whenever $(\tilde{X}, k)$ is an up-down node
and $H_k(\tilde{X}) \equiv 0$ otherwise.  It then follows from 
(\ref{afterAggregation}) that the right hand side of (\ref{removingInfiniteSum}) is bounded by $ V(X)$. From $ \overline \sigma_j f(X) < \infty $ it follows that we can choose the $f_m$ in such  a way that $V(X)$ approximates $\overline{\sigma}_jf(X)$. Therefore,  we have then established that the righ hand side of (\ref{removingInfiniteSum}) is bounded by $ \overline{\sigma}_j(X)$.

\vspace{.1in}
Next we establish the inequality $\leq $ in (\ref{removingInfiniteSum}).   Towards this end, assume there are functions $V:\mathcal{X}_{(X,j)}\rightarrow \mathbb{R}$  with finite maturity $j$ (i.e. $V(X)= V(X_0, \ldots, X_j)$)  and $H=(H_i)_{i=j,\ldots, n_f-1}$ non-anticipative such that 
\begin{equation} \nonumber 
f(\hat{X}) \leq V(X_0,\ldots, X_j)+\sum_{i=j}^{n_f-1} H_i(\hat{X})\Delta_i (\hat{X})~~\mbox{for any}~~\hat{X} \in \mathcal{X}_{(X,j)} \setminus \mathcal{N}.
\end{equation}
Therefore, if we let $g \equiv \infty ~{\bf 1}_{\mathcal{N} \cap \mathcal{X}_{(X,j)}}$,
\begin{equation}  \label{aeBound2}
f(\hat{X}) \leq V(X_0,\ldots,S_j)+\sum_{i=j}^{n_f-1} H_i(\hat{X})\Delta_i (\hat{X})+ g(\hat{X})~~\mbox{for any}~~\hat{X} \in \mathcal{X}_{(XX,j)}.
\end{equation}
By applying $\overline{\sigma}_j$ to  both sides of the inequality  (\ref{aeBound2}) and noticing that $\overline{I}_j g(X)=0$ follows from the countable subadditivity property of $\overline{I}_j$ and the fact that $\mathcal{N}$ is a null set, it follows that
\begin{equation}  \nonumber 
\overline{\sigma}_j f(X) \leq V(X_0,\ldots, X_j),
\end{equation}
where we have also used the inequality $\overline{\sigma}_j \leq \overline{I}_j$ on non-negative functions. The above inequality in
turn implies that $\overline{\sigma}_j f(X)$ is smaller or equal to the right hand side of (\ref{removingInfiniteSum}).    
\end{proof}


\subsection{Superhedging/Underhedging Pricing Algorithm}\label{sec:Pricing} 

We recall our use of capitalized letters, e.g. $X$, to denote model variables,
this is in contrast to observable quantities which are not capitalized, e.g. $x$.
This section, as well as follow up sections, concentrates in superhedging computations, in particular,  we assume the multidimensional trajectory set has already been constructed and so we will dispense with the {\it additional} coordinates. Therefore, for convenience, trajectories will be denoted by $X = (X_i)_{i\geq 0} = ((X^1_i, X^2_i))_{i \geq 0}$ for the purposes of what remains of the present appendix (i.e.,  for simplicity,  we are neglecting to include the additional variables). As we have already explained,
we construct trajectory sets for $d=2$ but the superhedging is only one-dimensional. This property of our approach is made explicit in the present section and, to avoid missunderstandings, we will introduce slightly different notation.
In particular the quantity corresponding to $\overline{\sigma}_j f$, when  performing $1$-dimensional computations but relying on a $2$-dimensional context,  will be denoted by $\overline{\sigma}^1_j f$ (see explicit definition in (\ref{sigma_X_def}) below).

In this section we introduce the convenient notation  $F(X)$ that represents a ``payoff" to be superhedged.  Also, as indicated above,  
$\overline{\sigma}^1_iF(X)$ 
will denote the superhedging price of $F$ at node $(X,i)$ but defined by portfolios trading only with the first asset.
The quantities $\overline{U}_iF(X)$ and $\underline{U}_iF(X)$ (see Definition \ref{dynamicBounds})  for $i\geq 0$ give an explicit dynamic programming formulation to calculate $\overline{\sigma}^1_iF(X)$ and $\underline{\sigma}^1_iF(X)$ respectively.

Once we have built our trajectory set $\mathcal{X}$, we proceed to superhedge one asset, which we may designate as the target asset and we denote  it by  $X^2$ relative to the asset $X^1$ (although we are free to reverse the roles when performing numerical experiments). 
Define $F:\mathcal{X}\rightarrow \mathbb{R}_+$ to be a function with finite maturity;  this is the \textit{payoff}  function which we aim to superhedge, 
for a trajectory $ X \in\mathcal{X}$:
\begin{equation} \nonumber 
F(X) \equiv X^2_{N(X)},
\end{equation}
where $N(X)$ is the terminal rebalance number for the trajectory $X$.
We will rely on the superhedging algorithm from \cite{degano} which is $1$-dimensional i.e. the trading uses a single asset (plus the numeraire) and for this reason we will need to introduce
some notation to account for this fact. Notice that one could trade on two assets to superhedge a third one (and so forth for higher dimensions), this will require that we construct a trajectory set 
for three traded coordinates and extend the results from \cite{degano} in order to handle higher dimensions (which is possible but it would require a separate work).

For a node $(X,j)$ and a general $F:\mathcal{X}\rightarrow\mathbb{R}$, define
\begin{equation} \label{sigma_X_def} 
\overline{\sigma}^1_j F(X) \equiv  \inf \left\{\sum_{m \geq 0}V^m: ~~F \leq  \sum_{m \geq 0} f_m \;\;\mbox{on}\; \Xe_{(X,j)}\right\},
\end{equation}
where
 $f_m(\hat{X})\equiv V^m+\sum_{i=j}^{n_m-1}H_i(\hat{X})(\hat{X}_{i+1}^1-\hat{X}_i^1)$ for $~m \geq 0,$ and $ V^m+\sum_{i=j}^{n-1}H_i(\hat{X})(\hat{X}_{i+1}^1-\hat{X}_i^1)\geq 0$ for all $m\geq 1$, $n \geq j$ and $\hat{X}\in\mathcal{X}_{(X,j)}$. Define also $\underline{\sigma}^1_j F(X) \equiv -\overline{\sigma}^1_j(-F) (X)$. 

Clearly, if we wish to price $X^1$ in terms of $X^2$, we simply reverse the corresponding indexes in the definition. The definition in display  (\ref{sigma_X_def}) is analogous to the definition of $\overline{\sigma}_j f$ in Definition \ref{cond_integ_def} with the following difference: in the present section $X=\{X_i= (X_i^1, X_i^2)\}_{i\geq 0}$ is a trajectory where $X_i$ contains two assets. As a consequence, we would ordinarily utilize both asset variables to superhedge a payoff, however, given that our goal is to price with only a single asset, we always hold 0 amount of the second asset. The function $H_i:\mathcal{X}_{(X,j)}\rightarrow \mathbb{R}$ in definition (\ref{sigma_X_def}) represents the number of shares of $X^1$ only and hence $f^m$ represents the value of a portfolio containing only shares of $X^1.$\\
 



It should be clear that we can consider a $1$-dimensional trajectory set $\tilde{\mathcal{X}}$ extracted from a $2$-dimensional trajectory set $\mathcal{X}$ as follows: for
each $X= (X^1_i, X^2_i)_{i \geq 0} \in \mathcal{X}$, 
set $\tilde{X} = (\tilde{X}_i)_{i \geq 0} \equiv (X^1_i)_{i \geq 0}$, $N(\tilde{X}) \equiv N(X)$  and define $f(\tilde{X})= F(X)= X^2_{N(X)}$ we then have 
\begin{equation} \nonumber 
\overline{\sigma}_jf(\tilde{X})= \overline{\sigma}^1_jF(X)
\end{equation}
 (where $\overline{\sigma}_jf(\tilde{X})$ is as in Definition \ref{cond_integ_def} with $d=1$ and relative to the trajectory set  $\tilde{\mathcal{X}}$) and so we can apply Theorem \ref{thm:sigma_is_correct} in order to evaluate $\overline{\sigma}^1_jF(X)$.
To sum up, we extract a $1$-dimensional trajectory set out of the originally built
$2$-dimensional trajectory set and in this way we have available theoretical results for the $1$-dimensional case. One then needs to check the availability of the hypotheses required for the application for the said theorem. This is a straightforward task, in particular simple sufficient conditions for the validity of $(L)-a.e.$ are provided in Corollary C.3  \cite{bender3}.

$\overline{\sigma}^1_iF(X)$ 
denotes the minimum amount of capital required, conditionally at node $(X,i)$, to superhedge the value of the trajectory $X^2$  at terminal time $N(X)$, with $0\leq i\leq N(X)$, using the available  trajectories of the single asset $X^1$. 
As usual, an analogous interpretation is available for $\underline{\sigma}^1_iF(X)\equiv -\overline{\sigma}^1_i(-F)(X)$
in order to underhedge asset $X^2$.  
We call $\underline{\sigma}^1_iF(X)$ and $~\overline{\sigma}^1_iF(X)$  
 \textit{price bounds} at the  node $(X,i)$; the following inequality holds $\underline{\sigma}^1_iF(X)= \underline{\sigma}_if(\tilde{X})  \leq \overline{\sigma}_if(\tilde{X})= \overline{\sigma}^1_iF(X)$ whenever $(L_{(\tilde{X},i)})$ holds at node $(\tilde{X},i)$. On the other hand, if $(\tilde{X},i)$ were a type II arbitrage node we would then have $-\infty=\overline{\sigma}^1_iF(X)<\underline{\sigma}^1_iF(X)=\infty.$ As a side remark we mention that, given our methodology to construct trajectory sets (in particular the characterization of the empirical set $N_E$), we do not expect to encounter Type I arbitrage nodes, nor do we encounter Type II 0-neutral nodes due to the discrete nature of our models.

\begin{corollary}  \label{1DimensionalInBetweenBounds}
Consider the same setting and assumption as in Proposition \ref{basicInequalities} above, then:
\begin{equation} \label{priceBetweenBoundsOneDimensional}
\underline{\sigma}^1_j X^2_{N(X)} \leq X^2_j \leq \overline{\sigma}^1_j X^2_{N(X)}. 
\end{equation}
\end{corollary} 
 \begin{proof}
Notice that $\overline{\sigma}_j X^2_{N(X)} \leq \overline{\sigma}^1_j X^2_{N(X)}$ as the left hand side is defined as an infimum over a larger set (i.e. portfolios with two tradable assets while the righ hand side is defined my means of portfolios defined with a single tradable asset). The latter inequality gives $\underline{\sigma}^1_j X^2_{N(X)} \leq \underline{\sigma}_j X^2_{N(X)}$. Therefore (\ref{priceBetweenBoundsOneDimensional})  follows from (\ref{eqPriceInBetweenBounds}).
\end{proof}
\begin{remark}  \label{validityOfL}
Proposition \ref{basicInequalities} and Corollary \ref{1DimensionalInBetweenBounds} both require the validity of 
$(L_{(X,j)})$ in a $2$-dimensional sense. On the other hand,  results establishing the validity of $(L_{(X,j)})$ and $(L)$-a.e. in
\cite{bender3} are $1$-dimensional.  So one rightfully inquires
about results delivering the validity of $(L_{(X,j)})$ and $(L)$-a.e. in the $2$-dimensional case; notice that $2$-dimensional results will imply immediately $1$-dimensional versions of those results. We claim that the weak sufficient conditions for 
$(L_{(X,j)})$ and $(L)$-a.e. in
\cite{bender3} can be extended to the multidimensional case, presenting these results will take substantial space and, for this reason, we will not embark on such work. We remark that the fundamental property, in its weaker version, is  $(L)$-a.e.  and results  in
\cite{bender3} do provide access to our Theorem \ref{thm:sigma_is_correct}
a key result for the present paper. Working under the hypothesis 
$(L)$-a.e.,  $1$-dimensional or $2$-dimensional (whatever is required for the result at hand), is the key that allows to include arbitrage nodes in our trajectory sets. In particular, under
 $(L)$-a.e., $2$ dimensional version, (\ref{priceBetweenBoundsOneDimensional}) will only be available
 a.e. 
\end{remark}



Reference \cite{degano}, under special conditions (discussed below in Section \ref{subsec:Ignoring_null_sets}), provides a rigorous algorithm to evaluate  the quantity $\overline{V}_j f(X)$ appearing in the right hand side of display (\ref{removingInfiniteSum})
in Theorem \ref{thm:sigma_is_correct} at node $(X,i)$, 
it achieves this goal by introducing intermediate
quantities, namely  $\overline{U}_iF(X)$. The procedure is a dynamic programming
algorithm which begins with the evaluation of the payoff function at maturity, i.e. at the final index $n\equiv N(X)$, i.e. $\overline{U}_nF(X) \equiv F(X)$  and proceeds to evaluate $\overline{U}_{i}F(X)$ backwards recursively over all nodes  $(X, i)$, for all $0 \leq i \leq n-1$. One can then prove that under general hypotheses  
$\overline{U}_{i}F= \overline{V}_{i}F$ for all $0 \leq i \leq n$, for example see \cite{degano}. 

The following inductive definition gives the basic dynamic programming formulation to compute $\overline{V}_0 F=\overline{U}_0F$ (see Definition 9 from \cite{degano}). 

\begin{definition}[Dynamic Bounds] \label{dynamicBounds}
For a given
 $X\in \mathcal{X}$, and
$0\le i \le n= N(X)$ set
\begin{equation}  \label{DynamicBound}
\overline{U}_i F(X)= \left\lbrace \begin{array}{lcc}
\inf\limits_{H\in\mathcal{H}} \sup\limits_{\hat{X} \in \mathcal{X}_{(X,i)}}
[\overline{U}_{i+1}F(\hat{X}) - H_i(X)
\Delta_i \hat{X}^1] &\mbox{if}& 0\le i< n,\\
F(X) &\mbox{if}& i=n,\\
0 &\mbox{if}& i>n,
\end{array}\right.
\end{equation}
where $\Delta_i\hat{X}^1=\hat{X}_{i+1}^1-\hat{X}^1_i$ and $\hat{X}\in\mathcal{X}_{(X,i)}$. 
Also define $\underline{U}_iF(X) \equiv -\overline{U}_i(-F)(X)$.
\end{definition}

The quantities $\underline{U}_iF(X),~\overline{U}_i F(X)$ are the mininimum/maximum price bounds for a one-step virtual market defined on $\mathcal{X}_{(X,i)}$ for the payoff $\overline{U}_{i+1}F(X)$.

For each trajectory $X\in\mathcal{X}$, the chosen $H=(H_i)_{i\geq 0}$ satisfying (\ref{DynamicBound}) is known as the \emph{superhedging portfolio} and then used to superhedge $X^2_{N(X)}$ (see Section \ref{sec:profitAndLossAnalysis} for examples). Notice that Definition \ref{dynamicBounds} represents superhedging by relying on  one-dimensional trajectories (as in \cite{degano}).

The practical
obstruction for a full generalization to the multidimensional case  is in the extension of the 
Convex Envelope Algorithm (see Section 4 of \cite{degano}) which actually computes the quantities $\overline{U}_iF$, to higher dimensions. This extension is an open problem and it is the main reason
why we have restricted ourselves to $d=1$-dimensional portfolio trading in this work even though our methodology
can produce $d$-dimensional trajectory models.

\subsection{Ignoring Null Sets Through Backward Recursion}\label{subsec:Ignoring_null_sets}
The above results and discussion rely on \cite{degano}, that reference  requires 
that $(L_{(X, j)})$ is valid at all nodes $(X,j)$. Notice that this hypothesis is in a $1$-dimensional sense and that 
the superhedging operator, in the said reference, is defined by means of simple portfolios (i.e. it does not use a countable family of simple portfolios). The assumption on the validity of  $(L_{(X, j)})$ at all nodes $(X,j)$ will indeed be uphold 
in a trajectory set where all nodes are either arbitrage-free or of type I arbitrage. 
Next  we explain, informally,  how to extend the results in \cite{degano} to a general case containing nodes of type II arbitrage as well. As already anticipated, this is a needed extension as type II arbitrage nodes
may appear during the construction of our trajectory sets (in contrast to, for example \cite{crisci}, where type II nodes were handled by artificially converting them to type I nodes via a flat trajectory).

Under the hypothesis $(L)-a.e.$, from Theorem \ref{thm:sigma_is_correct}, if we define $\mathcal{X}' \equiv \mathcal{X} \setminus \mathcal{X}$, we can then see, using results from \cite{bender3}, that  $(L_{(X', j)})$ holds at all nodes $(X',j)$, $X' \in \mathcal{X}'$,
$j \geq 0$. It then follows from \cite{degano} that $\overline{V}_j(X')= \overline{U}_j(X')$ and so, from Theorem \ref{thm:sigma_is_correct} we will have access to $\overline{\sigma}_j f(X')$, $X' \in \mathcal{X}'$. Notice also that
if $(X,j)$ is a type II arbitrage node we will have $\overline{\sigma}_j f(X) = - \infty$ for any function $f$ on $\mathcal{X}$. It then remains to describe the obvious modification to evaluate $\overline{U}_j$ under the presence of type II arbitrage nodes.

We now describe the practical modification to (\ref{DynamicBound}) which allows us to handle null subsets of trajectories; more specifically, those subsets containing trajectories which pass through a type II arbitrage node. As nodes are created using our forward recursive algorithm (which involves dynamic pruning), they are then tested for arbitrage. The modification to the pricing algorithm described by equation (\ref{DynamicBound}), will disregard/remove all Type II arbitrage children nodes $(\hat{X},i+1)$ in the inner supremum. Care must be taken, as the procedure just described may remove too many children leading to a convex hull violating the first item in Proposition \ref{characterizationOfNodes} at the parent node $(X,i)$. Hence, we test for the property 
$(L_{(X,i)})$ (in one dimension, given that we are pricing with a single asset) at each parent node  during the  backward recursive pricing algorithm after removing all arbitrage type II children nodes. 

 Our superhedging and underhedging prices are not affected by the removal of type II arbitrage nodes, i.e. the removal of a null set (containing trajectories passing through a type II node) has no effect on prices.
Given this more general pricing algorithm, it is possible to get $\overline{\sigma}_0F(X)=-\infty$, meaning either that the initial node $(X,0)$ is a type II arbitrage node (which only occurs if pruning constraints severely limited the number of available children), or there weren't enough 0-neutral children nodes to price using (\ref{DynamicBound}). This occurs especially with a small empirical set, such as $|N_E|=3$. Small empirical sets $N_E$ or small sampling sizes for $N_E$ were found to produce many arbitrage nodes after pruning, which could often lead to degenerative (i.e. $\pm\infty$) price bounds at $(X,0)$. We should note that this phenomenon is sensitive to pruning, meaning tighter bounds could lead to degenerative prices.



\section{Appendix B. Additional Pruning Constraints}\label{additional_pruning_constraints}

The additional variables (see their formal introduction in Definition \ref{trajectories}) are used for pruning during the forward pass of the recursive algorithm that constructs trajectories. They have no further use in our approach.

 In order to limit the amount of trajectories, and to have them reflect more closely past historical data, we utilize  \emph{pruning constraints} or \emph{pruning functions}. 
 
The first pruning constraint monitors the relative vector norms and does not involve the additional coordinates, while the other constraints come in pairs, e.g., number of $\delta$-escapes constrained by time and vice versa, accumulated variation constrained by time and vice versa, as well as accumulated variation constrained by the number of $\delta$-escapes and vice versa (plus other possibilities).

Each constraint is a pair of functions; a maximum and minimum quantity denoted by $^*$ and $_*$ respectively. Each function depends on a particular chart $x$ and requires the investor to move over all historical windows $\It\in I$, hence $x$ and $I$ appear as arguments for each constraint. The third argument for each constraint represents either time $\rho$, number of $\delta$-escapes $i$ or variation $w$.  

When the pruning constraint is a function of time $\rho$, the maximum and minimum are taken over the  set of all (historical) windows i.e. $\It \in I$ .   In the case when the pruning constraint is a function of the number of $\delta$-escapes $i\in\mathbb{Z}_+,$ there may be some windows which admit less than $i$ $\delta$-escapes; hence we maximize/minimize over only the windows which pick up at least $i$ $\delta$-escapes, i.e. we take the maximum/minimum over $\It\in I_i$ where $I_i\subseteq I$ is defined as:
\begin{equation} \label{admissibleSetOfWindows}
I_i=\{\It\in I: N(x,\It)\geq i\},
\end{equation}
where the notation  $N(x,\It)$ was introduced at the end of Section \ref{deltaEscapeTimes}. For the case when the pruning constraint is a function of variation $w$ (introduced in Section \ref{sec:Variation}), different windows admit different ranges for the values of variation. Therefore, for a particular $w\in\mathbb{Z}_+$ we maximize/minimize over only the windows which admit exactly the variation $w$, i.e. we maximize/minimize over $\It\in I_w$ where $I_w\subseteq I$ is defined as:
\begin{equation} \label{admissibleSetOfWindows2}
I_w=\{\It\in I: ~\exists ~\rho\in\{0,\Delta, \hdots, M_T\Delta\}:~ w(x,\It,\rho)=w\},
\end{equation}
where $w(x,\It,\rho)$ is given by (\ref{Pruning2_1}) below.

In this section, for purely pedagogical reasons,  we classify pruning constraints into Type 0, Type I and Type II pruning constraints, as in \cite{crisci}. The classification is motivated by the types of variables encountered; namely the number of $\delta$-escapes $i,$ the time $\rho$ and accumulated variation $w.$ Type I pruning constraints don't require accumulated variation $w,$ while Type 0 pruning constraints don't require neither the accumulated variation $w,$ nor the elapsed time $\rho.$ When building our models, we utilize all pruning constraints, but we emphasize the Type 0, Type I and Type II pruning constraints here, since the investor may wish to limit the type of variables in their model.
\subsection{Type 0 Pruning Constraint}\label{subsec::type0model}

The following two constraints help to limit the amount of fluctuation of trajectory asset values.

\begin{definition}[Historical Maximum and Minimum Relative Normed Changes]\label{def::type0pruningconstraint}
For a given chart $x,$ time interval $\It\in I$ and portfolio rebalances times $t(\It)=\{t_i\}_{0\leq i\leq N}$, define:
\begin{equation} \nonumber 
X_{norm}(x,I_{t_0},i)=\frac{\|x(t_{i})-x(t_0)\|}{\|x(t_0)\|}    
\end{equation}
for $0\leq i\leq N$ where $N=N(x,\It)$.

Then the corresponding Historical Maximum and Minimum Relative Normed Changes over the set of time intervals $I$ is given by:

\begin{equation}\nonumber
X^*(x,I,i) = \max_{I_{t_o}\in I_i}\:\: X_{norm}(x,I_{t_0},i),\quad X_*(x,I,i) = \min_{I_{t_o}\in I_i}\:\: X_{norm}(x,I_{t_0},i)
\end{equation}

for $0\leq i \leq i^*$ and $I_i$ as introduced in (\ref{admissibleSetOfWindows}).%
\end{definition}

\subsection{Type I Pruning Constraints} \label{subsec::typeImodel}

\begin{definition}[Historical Maximum and Minimum Elapsed Time]\label{def::typeIIpruningconstraint_T_vs_i}
For a given chart $x,$ time interval $\It$ and portfolio rebalances times $t(\It)=\{t_i\}_{0\leq i\leq N}$, define the elapsed time to be: 
\begin{equation} \nonumber 
 T(x,\It,i)=t_i-t_0   
\end{equation} for $\:\:0\leq i\leq N$ where $N=N(x,\It).$
\\

Then the Historical Maximum and Minimum Elapsed Time are defined as:
\begin{align}\label{eqn::typeIIpruningconstraint_T_vs_i}
T^*(x,I,i) &= \max_{I_{t_o}\in I_i} \:\: T(x,\It,i),\quad 
\nonumber T_*(x,I,i) = \min_{I_{t_o}\in I_i}\:\: T(x,\It,i),
\end{align}
where $0\leq i \leq i^*$  and $I_i$ as introduced in (\ref{admissibleSetOfWindows}).
\end{definition}

\subsection{Type II Pruning Constraints}\label{subsec::typeIImodel}
Type II pruning constraints incorporate the variation $w$ and time $\rho$
and so this section expands Definition \ref{def::typeIpruningconstraint_N_vs_T}.  All related definitions involving variation  are derived from the original definition in  Equation (\ref{w}). $W^*(x,I,\rho)$, $W_*(x,I,\rho)$ and $W^*(x,I,i),$  $W_*(x,I,i)$ will denote the worst case historical values of variation, at time $\rho$ and rebalance $i$ respectively , and will be used to restrict variation.\\

$N^*(x,I,w), N_*(x,I,w)$ and $T^*(x,I,w), T_*(x,I,w)$ pair the historic number of $\delta$-escapes and $\delta$-escape times to accumulated variation $w$; they are  the worst case historical values of the number of rebalances and times respectively (given variation $w$), and limit how  these quantities may evolve given some accumulated variation in the future.

\begin{definition}[Historical Maximum and Minimum Vector Variation at Time $\rho$]\label{def::typeIIpruningconstraints_W_vs_T}
For a given chart $x,$ time interval $\It$ and portfolio rebalances times $t(\It)=\{t_i\}_{0\leq i\leq N}$, define the  accumulated vector variation at time $\rho$ as: \begin{equation}\label{Pruning2_1}
w(x,\It,\rho)=w(x,\It,\rho-\Delta)+\lvert k_{\rho}^1-k_{\rho-\Delta}^1\rvert+\lvert k_{\rho}^2-k_{\rho-\Delta}^2\rvert,    
\end{equation}

for $\rho\in\{\Delta, \hdots, M_T\Delta\}$ and $w(x,\It,0)=0$. Notice $w(x,\It,\rho)\in\mathbb{Z}_+.$\\

The Historical Maximum and Minimum Vector Variation at time $\rho\in\{0,\Delta, \hdots, M_T\Delta\}$ is then given by:
\begin{equation}\nonumber
W^*(x,I,\rho) = \max_{I_{t_o}\in I}\:\:   w(x,\It,\rho), \quad W_*(x,I,\rho) = \min_{I_{t_o}\in I}\:\:  w(x,\It,\rho).
\end{equation}

\end{definition}

\begin{definition}[Historical Maximum and Minimum Vector Variation at Rebalance $i$]\label{def::typeIIpruningconstraints_W_vs_i}
For a given chart $x,$ time interval $\It$ and set of portfolio rebalances times $t(\It)=\{t_i\}_{0\leq i\leq N}$ define the vector variation at rebalance $i\in\{0,  \hdots N\}$ in the following way: let  $t_{i-1}=u\Delta$ and $t_i=v\Delta$ for some $u, v\in\mathbb{Z}_+$. Then    \begin{equation}\label{eq:wi_single}
w(x,\It,i)=w(x,\It,i-1)+\sum_{j=u}^{(v-1)}\lvert k_{(j+1)\Delta}^1-k_{j\Delta}^1\rvert+\lvert k_{(j+1)\Delta}^2-k_{j\Delta}^2\rvert,
\end{equation} for $i\in\{1, \hdots, N\} $  
 where $N=N(x,\It)$ and $w(x, \It,0)=0.$\\

Then the Historical Maximum and Minimum Vector Variation at Rebalance $i\in\{0, \hdots, i^*\}$ is given by:

\begin{equation}\nonumber
W^*(x,I,i) = \max_{I_{t_o}\in I_i}\:\:  w(x,\It,i), \quad W_*(x,I,i) = \min_{I_{t_o}\in I_i}\:\: w(x,\It,i).\\
\end{equation}
\end{definition}

The final two kinds of pruning constraints use the accumulated variation as the variable, which needs to first be calculated on its own through (\ref{Pruning2_1}).
\begin{definition}[Historical Maximum and Minimum Number of $\delta$-movements (at accumulated vector variation $w$)]\label{def::typeIIpruningconstraint_N_vs_W}

For a given chart $x$ and time interval $\It$, portfolio rebalancing times $t(\It)=\{t_i\}_{0\leq i\leq N}$, we have $N(x,\It, \rho)$ defined through equation (\ref{eqn::typeIpruningconstraint_N_vs_T}), For each $\rho\in\{0, \Delta, \hdots, M_T\Delta\}$ let $w=w(x,\It,\rho).$ 
Then define 
\begin{equation}\label{Nw_single_eq}
N(x,\It,w)=N(x,\It,\rho)    
\end{equation}
where $\rho$ satisfies $w(x,\It,\rho)=w,$  defined through equation (\ref{Pruning2_1}).\\



Then, for accumulated vector variation $w\in\mathbb{Z}_+$ satisfying $w=(w,\It,\rho)$ for some $\It\in I$ and $\rho\in \{0, \Delta, \hdots, M_T\Delta\},$ the Historical Maximum and Minimum Number of $\delta$-movements is given by:
\begin{equation}\nonumber
\begin{split}
N^*(x,I,w) &= \max_{\substack{I_{t_o}\in I_w,}}\:\: N(x,\It,w),\quad N_*(x,I,w) = \min_{\substack{I_{t_o}\in I_w}}\:\: N(x,\It,w),
\end{split}
\end{equation}
where $I_w$ is given by (\ref{admissibleSetOfWindows2}).
\end{definition}

\begin{definition}[Historical Maximum and Minimum Elapsed Time (at accumulated variation $w$)]\label{def::typeIIpruningconstraint_T_vs_W}
For a given chart $x$ and time interval $\It$, portfolio rebalancing times $\{t_i\}_{0\leq i\leq N}$, define for each $\rho\in\{0, \Delta, \hdots, M_T\Delta\}$ the accumulated variation  $w=w(x,\It,\rho).$ Then define
\begin{equation}\label{Tw_single_eq}
    T(x,\It,w)=\rho,
\end{equation}
where $\rho$ satisfies $w(x,\It,\rho)=w,$  defined through equation (\ref{Pruning2_1}).

Then, for accumulated vector variation $w\in\mathbb{Z}_+$ satisfying $w=(w,\It,\rho)$ for some $\It\in I$ and $\rho\in \{0, \Delta, \hdots, M_T\Delta\},$ the Historical Maximum and Minimum Elapsed Time is given by:
\begin{equation}\nonumber
\begin{split}
T^*(x,I,w) = \max_{\substack{I_{t_o}\in I_w}} \:\: T(x, \It,w),\quad  T_*(x,I,w) = \min_{\substack{I_{t_o}\in I_w}}\:\: T(x,\It,w),
\end{split}
\end{equation}
where $I_w$ is given by (\ref{admissibleSetOfWindows2}).
\end{definition}

We mention a couple more implicit restrictions on our trajectory sets. Historically, all observed time intervals $\It=[t_0, t_0+T]$ for $t_0=-T, -2T-1, \hdots$  are of length $T$. Therefore, we restrict our trajectories to evolve for no longer than time $T$ into the future; essentially we are simulating one time interval's length of time into the future.\\

Let us introduce \begin{equation*}\nonumber
w^*(x,\It)\equiv w(x,\It,\rho=M_T\Delta)
\end{equation*}
be the maximum accumulated variation over a particular time window $\It$ and let
\begin{equation} \nonumber 
w^*\equiv w^*(x,I)=\max\{w^*(x,\It):\It\in I\}    
\end{equation} be the maximum accumulated variation over all time windows. When we later build trajectory sets, it is possible that we obtain an accumulated variation $W_i>w^*$ (see Section \ref{modelSpecification} 
for the definition of the model variable $W_i$). Such a scenario never occurs historically, by definition of $w^*$, hence for any $w>w^*$  we set $N^*(x,I,w)=N_*(x,I,w)=0=T^*(x,I,w)=T_*(x,I,w).$\\

\end{appendices}

\section*{Acknowledgments}
S. E. Ferrando acknowledges partial financial support from NSERC for the completion of this project.



\begin{thebibliography}{999}


\bibitem[Bartl at al. (2020)]{bartl}
\newblock D. Bartl, M. Kupper, and A. Neufeld, 
\newblock \emph{Pathwise superhedging on prediction sets}, 
\newblock Finance Stoch., 24 (2020), pp. 215--248.


\bibitem[Bender et al. (2021)]{bender1} 
\newblock Bender C., Ferrando S.E. and Gonzalez A.L.,  
\newblock \emph{Conditional Non-Lattice Integration, Pricing and Superhedging}; 
\newblock  Revista de la Uni\'on Matem\'atica Argentina. Published online (early view, August 28, 2024);
{\tt https://doi.org/10.33044/revuma.4351}


\bibitem[Bender et al (2023)]{bender3} 
\newblock Bender C., Ferrando S.E., Gajewski K. and Gonzalez A.L.,  \newblock \emph{Superhedging Supermartingales}; 
{\tt arXiv:2312.14445 [math.PR]}  (2023).

\bibitem[Blanchard  and Carassus L. (2020)]{blanchard} 
\newblock Blanchard R. and Carassus L., 
\newblock \emph{No-arbitrage with multiple-priors in discrete time},
\newblock  Stochastic Process. Appl., 130 (2020), pp. 6657--6688.






\bibitem[Burkholder (1989)]{burkholder}  
\newblock Burkholder D.L., 
\newblock \emph{On the number of escapes of a martingale and its
geometrical significance}, 
\newblock Almost Everywhere Convergence, 
\newblock edited by Gerald A. Edgar and Louis Sucheston. Academic Press, New York (1989),  159--178.

\bibitem[Burzoni et al. (2019)]{burzoni} 
\newblock Burzoni M., Frittelli M., Hou Z., Maggis M. and Oblooj J., 
\newblock \emph{Pointwise arbitrage pricing theory in discrete time}, 
\newblock Math. Oper. Res., 44 (2019), 1035--1057.



\bibitem[Crisci (2019)]{crisci} 
\newblock Crisci  D., 
\newblock \emph{Trajectory Based Market Models for Two Stocks}, 
\newblock MSc. Thesis, Department of Mathematics, Ryerson  University,
Canada (2019). 
{\tt https://doi.org/10.32920/ryerson.14665113.v1}





\bibitem[Degano et al. (2018)]{degano} 
\newblock Degano, I. L., Ferrando, S. E. and  Gonz\'alez, A. L., 
\newblock \emph{Trajectory Based Market Models. Evaluation of Minmax Price Bounds}. 
\newblock Dynamics of Continuous, Discrete and Impulsive Systems Series B: Applications \& Algorithms. {\bf 25:2} (2018), 97--128.

\bibitem[Degano et al. (2022)]{degano2} 
\newblock Degano, I. L., Ferrando, S. E., and Gonz\'alez, A. L.,
\newblock \emph{No-Arbitrage Symmetries}. 
\newblock Acta Mathematica Scientia. {\bf 42} (2022), 1373--1402.


\bibitem[Ferrando and Gonzalez (2018)]{ferrando2} 
\newblock Ferrando S.E.  and Gonzalez A.L., 
\newblock \emph{Trajectorial martingale transforms.
Convergence and integration}. 
\newblock New York Journal of Mathematics, {\bf 24} (2018), 702--738.

\bibitem[Ferrando, Fleck et al. (2019)]{ferrando19} 
\newblock Ferrando, S. E., Fleck, A., Gonz\'alez, A. L. and Rubtsov, A., \newblock \emph{Trajectorial asset models with operational assumptions}. 
\newblock Quantitative Finance and Economics, {\bf 3:4} (2019), 661--708.

\bibitem[Ferrando,  Gonz\'alez, et al.  (2019)]{ferrando} 
\newblock Ferrando, S. E.,  Gonz\'alez, A. L., Degano, I. L. and Rahsepar, \newblock \emph{Trajectorial Market Models. Arbitrage and Pricing Intervals}.
\newblock Revista de la Uni\'on Matem\'atica Argentina. {\bf 60:1} (2019), 149--185.


\bibitem[F{\"o}llmer  and Schied  (2011)]{follmer} 
\newblock F{\"o}llmer, H. and Schied, A. 
\newblock Stochastic Finance: An Introduction in Discrete Time, 3rd Edition. \newblock De Gruyter (2011), Berlin.

\bibitem[Gajewski (2022)]{konrad} 
\newblock Gajewski  K., 
\newblock \emph{Non-Probabilisitic Supermartingales. Trajectorial Models for Two Stocks.}, 
\newblock PhD. Thesis, Department of Mathematics, Toronto Metropolitan University (2022). 




\bibitem[Liebrich et al. (2022)]{liebrich} 
\newblock Liebrich F-B, Maggis M. and Svindland G. 
\newblock \emph{Model Uncertainty: A Reverse Approach}.
\newblock SIAM Journal of Financial Mathematics, {\bf 13:3} (2022), 1230--1269.


\bibitem[Neveu (1975)]{neveu} 
\newblock Neveu, J.
\newblock Discrete-parameter Martingales. 
\newblock North-Holland (1975).





\bibitem[Shafer and Vovk (2019)]{shafer} 
\newblock Shafer G. and Vovk V.,
\newblock Game-Theoretic Foundations for Probability and Finance. 
\newblock Wiley  (2019).

\bibitem[Vecer (2011)]{vecer} 
\newblock Vecer, J.
\newblock Stochastic Finance. A Numeraire Approach, 
\newblock CRC Press (2011).




\bibitem[van der Vaart  and Wellner (1996)]{VdVW}
\newblock van der Vaart A. W.  and Wellner J. A., 
\newblock Weak Convergence and Empirical Processes. Springer 
\newblock Verlag (1996).

\end{thebibliography}
\end{document}